\documentclass[11pt,leqno]{article}
\usepackage{amssymb,amsfonts}
\usepackage{amsmath,amsthm,amsxtra}
\usepackage[all]{xy}
\usepackage{color}
\usepackage{verbatim}

\usepackage{enumerate}

\newcommand\bigcheck[1]{#1 \raise1ex\hbox{$\hspace{-1ex}{}^\vee$}}
\newcommand\sucheck[1]{#1 \raise0.5ex\hbox{$\hspace{-1ex}{}^\vee$}}

\makeatletter
\@addtoreset{equation}{section}
\makeatother

\newtheorem{theorem}{Theorem}[section]
\newtheorem{lemma}[theorem]{Lemma}

\newtheorem{proposition}[theorem]{Proposition}
\newtheorem*{lemma*}{Lemma}

\theoremstyle{definition}
\newtheorem{definition}[theorem]{Definition}

\theoremstyle{remark}
\newtheorem{remark}[theorem]{Remark}

\newcommand{\mc}[1]{{\mathcal #1}}
\newcommand{\mf}[1]{{\mathfrak #1}}
\newcommand{\mb}[1]{{\mathbb #1}}

\newcommand\tint{{\textstyle\int}}

\newcommand{\id}{{1 \mskip -5mu {\rm I}}}

\renewcommand{\tilde}{\widetilde}

\newcommand{\ad}{\mathop{\rm ad }}

\newcommand{\Mat}{\mathop{\rm Mat }}

\renewcommand{\ker}{\mathop{\rm Ker }}
\newcommand{\im}{\mathop{\rm Im }}

\newcommand{\Span}{\mathop{\rm Span }}
\newcommand{\ord}{\mathop{\rm ord }}
\newcommand{\dord}{\mathop{\rm dord }}

\newcommand{\ass}[1]{\stackrel{#1}{\longleftrightarrow}}

\definecolor{light}{gray}{.9}

\begin{document}


\title{Non-local Poisson structures and applications to the theory of integrable systems II}

\author{
Alberto De Sole
\thanks{Dept. of Math., University of Rome ``La Sapienza'',
Rome, Italy,
and IHES, France,
desole@mat.uniroma1.it 
},~~
Victor G. Kac
\thanks{Department of Math., M.I.T.,
Cambridge, MA 02139, USA.
and IHES, France,
kac@math.mit.edu
}~~
}

\maketitle


\begin{abstract}
\noindent 
We develop further the Lenard-Magri scheme of integrability for a pair of compatible 
non-local Poisson structures, which we discussed in Part I.
We apply this scheme to several such pairs,
proving thereby integrability of various evolution equations, as well as hyperbolic equations.
Some of these equations may be new.
\end{abstract}


\section{Introduction}
\label{sec:intro}

In Part I of this series of papers we laid grounds of a rigorous theory of non-local Poisson structures.
In particular we have developed the Lenard-Magri scheme of integrability
for a pair of compatible non-local Poisson structures.

In Section \ref{sec:2} of the present Part II
we made several improvements of this scheme,
which facilitate its application in concrete examples.
The concrete examples are studied in Sections \ref{sec:3}, \ref{sec:4} and \ref{sec:5}.
In Section \ref{sec:3} we consider three compatible scalar non-local Poisson structures:
$$
L_1=\partial\,,\,\, 
L_2=\partial^{-1}\,,\,\, 
L_3=u'\partial^{-1}\circ u'
\,(\text{ Sokolov } \cite{Sok84})\,, 
$$
and take for a compatible pair $(H,K)$ arbitrary linear combinations of these three structures:
$H=\sum_ia_iL_i$, $K=\sum_ib_iL_i$.
We study in detail for which values of the coefficients $a_i$ and $b_i$
the corresponding Lenard-Magri scheme is integrable.
This means that there exists an infinite sequence of Hamiltonian functionals $\tint h_n,\,n\in\mb Z_+$,
and Hamiltonian vector fields $P_n,\,n\in\mb Z_+$,
such that we have
\begin{equation}\label{1.1}
\tint 0\ass{H}P_0\ass{K}\tint h_0\ass{H} P_1\ass{K}\tint h_1\ass{H}\dots\,,
\end{equation}
and the spans of the $\tint h_n$'s and of the $P_n$'s are infinite dimensional
(the association relation $\tint f\ass{H}P$ is recalled in Definition \ref{20120124:def}).
This produces integrable equations $u_{t_n}=P_n$.

Furthermore we study when the infinite sequence \eqref{1.1} can be extended to the left.
The most interesting cases are those when the sequence is ``blocked''
at some step $P_{-n},\,n>0$, to the left.
This leads to some interesting integrable hyperbolic equations.
As a result, we get two such equations
\begin{equation}\label{20121020:eq8-intro}
u_{tx} =
e^{u}-\alpha e^{-u}
+\epsilon(e^{u}-\alpha e^{-u})_{xx}
\,,
\end{equation}
where $\alpha$ and $\epsilon$ are $0$ or $1$, and
\begin{equation}\label{20121020:eq10-intro}
u_{tx} =
u+(u^3)_{xx}
\,.
\end{equation}
Of course, in the case when $\epsilon=0$, equation \eqref{20121020:eq8-intro} is the Liouville
(respectively sinh-Gordon) equation if $\alpha=0$ (resp. $\alpha=1$).
But for $\epsilon=1$ equation \eqref{20121020:eq8-intro} seems to be new.
Equation \eqref{20121020:eq10-intro}, first discovered in \cite{SW02},
is called the short pulse equation.
Its integrability was proved in \cite{SS04}

In Section \ref{sec:4}
we study, in a similar way, two liner combinations of the compatible non-local Poisson structures
$$
L_1=
u'\partial^{-1}\circ u' 
\,\,,\,\,\,\,
L_2=
\partial^{-1}\circ u'\partial^{-1}\circ u'\partial^{-1}
\,\,(\text{ Dorfman } \cite{Dor93})
\,.
$$
As a result we (re)prove integrability of the Krichever-Novikov equation
$$
u_t=u_{xxx}-\frac32\frac{u_{xx}^2}{u_x}\,,
$$
and also, moving to the left, establish integrability of the following equation
$$
\Big(\frac{u_{tx}}{u_x}\Big)_x
=
\frac1{2u_x}+\gamma u_x
\,,
$$
where $\gamma$ is a constant.

Finally, in Section \ref{sec:5}
we study, in a similar way, three two-components non-local Poisson structures
that are used in the study of the non-linear Schroedinger equation (NLS), \cite{TF86,Dor93,BDSK09,DSK12}.
As a result, we establish integrability of the following generalization of NLS:
$$
i\psi_t=\psi_{xx}+\alpha\psi|\psi|^2+i\beta(\psi|\psi|^2)_x\,,
$$
where $\alpha$ and $\beta$ are constants
(NLS is obtained for $\beta=0$).

We wish to thank Alexander Mikhailov and Vladimir Sokolov for useful correspondence,
and IHP and IHES for their hospitality.


\section{The Lenard-Magri scheme of integrability}
\label{sec:2}

\subsection{Non-local Poisson structures on an algebra of differential functions}
\label{sec:2.1}


Let $R_\ell=\mb F [u_i^{(n)}\, |\, i \in I,n \in \mb Z_+]$ be the algebra of differential polynomials
in the $\ell$ variables $u_i,\,i\in I=\{1,\dots,\ell\}$,
with the derivation $\partial$ defined by $\partial (u_i^{(n)}) = u^{(n+1)}_i$.
The partial derivatives $\frac{\partial}{\partial u_i^{(n)}}$
are commuting derivations of $R_\ell$, and they satisfy the following 
commutation relations with $\partial$:
\begin{equation}\label{eq:0.4}
\left[  \frac{\partial}{\partial u_i^{(n)}}, \partial \right] = \frac{\partial}{\partial u_i^{(n-1)}}
\,\,\,\,
\text{ (the RHS is $0$ if $n=0$) }\,.
\end{equation}
Recall from \cite{BDSK09} that an \emph{algebra of differential functions} 
is a differential algebra extension $\mc V$ of $R_\ell$,
endowed with commuting derivations
$$
\frac{\partial}{\partial u_i^{(n)}}:\,\mc V\to\mc V
\,\,,\,\,\,\,
i \in I ,\,n \in \mb Z_+\,,
$$
extending the usual partial derivatives on $R_\ell$,
such that only a finite number of
$\frac{\partial f}{\partial u_i^{(n)}}$ are non-zero for each $f\in \mc V$, 
and such that the commutation rules \eqref{eq:0.4} hold on $\mc V$.

A \emph{field of differential functions} is a field
which is an algebra of differential functions.
In particular, it contains the field of rational functions in the variables $u_i^{(n)}$,
which is itself a field of differential functions.
Given a field of differential functions $\mc V$,
it is easy to see that it can be extended
by adding solutions of algebraic equations over $\mc V$,
and also by functions of the form $F(\varphi_1,\dots,\varphi_n)$,
where $F$ is an infinitely differentiable function in $n$ variables
and $\varphi_1,\dots,\varphi_n$ lie in $\mc V$.
(Note, though, that in general we cannot add to $\mc V$ solutions of linear differential equations.
For example, a solution of the equation $f'=fu'$ is $f=e^u$, which can be added,
while a non-zero solution of the equation $f'=fu$ can never be added,
due to simple differential order considerations).
We will use these facts in the examples further on.

Note that, if the algebra of differential functions is a domain,
then its field of fractions is a field of differential functions,
with the maps $\frac{\partial}{\partial u_i^{(n)}}$ defined in the obvious way.

We will assume throughout the paper that $\mc V$ is a field of differential functions
in the variables $u_1,\dots,u_\ell$.
We denote by $\mc C=\big\{c\in\mc V\,\big|\,\partial c=0\big\}\subset\mc V$ the subset 
of \emph{constants}, and by
$$
\mc F=\Big\{f\in\mc V\,\Big|\,\frac{\partial f}{\partial u_i^{(n)}}=0
\,\,\text{ for all } i\in I,n\in\mb Z_+\Big\}\subset\mc V
$$
the subset of \emph{quasiconstants}. 
It is easy to see that $\mc C\subset\mc F\subset\mc V$
is a tower of field extensions.

Given $f\in\mc V$ which is not a quasiconstant, we say that it has \emph{differential order} $N$,
and we write $\dord(f)=N$,
if $\frac{\partial f}{\partial u_i^{(N)}}\neq0$ for some $i\in I$,
and $\frac{\partial f}{\partial u_j^{(n)}}=0$ for every $j\in I$ and $n>N$.
By definition, the differential order of a non-zero quasiconstant is $-\infty$.
We let $\mc V_N$ be the subfield of elements of differential order at most $N$.
This gives an increasing sequence of subfields 
\begin{equation}\label{difford}
\mc C\subset\mc F\subset\mc V_0\subset\mc V_1\subset\dots\subset\mc V\,,
\end{equation}
such that $\partial\mc V_N\subset\mc V_{N+1}$.
It is easy to show, using \eqref{eq:0.4}, that
\begin{equation}\label{20120907:eq1}
\partial\mc V\cap\mc V_N=\partial\mc V_{N-1}
\,\text{ for } N\geq 1\,,
\,\,\text{ and }\,\,
\partial\mc V\cap\mc V_0=\partial\mc F
\,.
\end{equation}
Given an arbitrary $k\times \ell$-matrix $A$ with entries in $\mc V$,
we define its differential order, denoted by $\dord(A)$,
as the maximal differential order of all its entries.

Given a matrix differential operator $D=\sum_{i=0}^n A_i\partial^i\in\Mat_{k\times\ell}\mc V[\partial]$,
we define its \emph{differential order} as
\begin{equation}\label{20120910:eq1}
\dord(D)=\max\{\dord(A_1),\dots,\dord(A_n)\}\,,
\end{equation}
which should not be confused with its \emph{order},
defined as 
\begin{equation}\label{20120910:eq2}
|D|=n \,\,\text{ if }\, A_n\neq0\,.
\end{equation}
(Note that the notion of order carries over to matrix pseudofferential operators,
while the differential order is not defined in general).
\begin{lemma}\label{20120910:lem1}
Let $D\in\Mat_{k\times\ell}\mc V[\partial]$
be a matrix  differential operator over $\mc V$
and let $F\in\mc V^\ell$.
Then:
\begin{enumerate}[(a)]
\item
$\dord(DF)\leq\max\{\dord(D),\dord(F)+|D|\}$.
\item
If $D$ has non-degenerate leading coefficient
and it satisfies $\dord(F)+|D|>\dord(D)$, then $\dord(DF)=\dord(F)+|D|$.
\item
If $D$ has non-degenerate leading coefficient
and it satisfies $\dord(DF)>\dord(D)$, then $\dord(DF)=\dord(F)+|D|$.
\end{enumerate}
\end{lemma}
\begin{proof}
Let $D=\sum_{s=0}^n A_s\partial^s$.
Clearly, for $f\in\mc V$ and $s\in\mb Z_+$, we have $\dord(f^{(i)})=\dord(f)+i$.
Hence, If  $h>\max\{\dord(D),\dord(F)+|D|\}$, we have
$$
\frac{\partial}{\partial u^{(h)}}(DF)_i=
\sum_{j=1}^\ell\sum_{s=0}^n \frac{\partial}{\partial u^{(h)}} (A_s)_{ij}F^{(s)}_j=0\,,
$$
for every $i=1,\dots,k$, proving part (a).
Furthermore, 
if $|D|=n$ and $\dord(F)+n>\dord|D|$, we can use \eqref{eq:0.4} to get
$$
\begin{array}{l}
\displaystyle{
\frac{\partial}{\partial u^{(\dord(F)+n)}}(DF)_i
=\sum_{j=1}^\ell\sum_{s=0}^n \frac{\partial}{\partial u^{(\dord(F)+n)}} (A_s)_{ij} F^{(s)}_j
} \\
\displaystyle{
=\sum_{j=1}^\ell (A_n)_{ij} \frac{\partial F^{(n)}_j}{\partial u^{(\dord(F)+n)}}
=\sum_{j=1}^\ell (A_n)_{ij} \frac{\partial F_j}{\partial u^{(\dord(F))}}
\,.
}
\end{array}
$$
Since, by assumption,
the leading coefficient $A_n\in\Mat_{\ell\times\ell}\mc V$ of $D$ is non-degenerate,
the RHS above is non-zero for some $i$.
Hence, 
$\dord(DF)=\dord(F)+n$, proving part (b).
Part (c) follows from parts (a) and (b).
\end{proof}

For $f\in\mc V$, as usual we denote by $\tint f$
the image of $f$ in the quotient space $\mc V/\partial\mc V$.
Recall that, by \eqref{eq:0.4}, we have a well-defined variational derivative
$\frac{\delta}{\delta u}:\,\mc V/\partial\mc V\to\mc V^{\ell}$,
given by
$$
\frac{\delta\tint f}{\delta u_i}=\sum_{n\in\mb Z_+}(-\partial)^n\frac{\partial f}{\partial u_i^{(n)}},\,i\in I\,.
$$
\begin{remark}\label{20120909:rem1}
Recall that $\ker\big(\frac\delta{\delta u}\big)\supset\mc F+\partial\mc V$,
and, by \cite[Prop.1.5]{BDSK09}, equality holds if $\mc V$ is normal.
\end{remark}


Denote by $\mc V((\partial^{-1}))$ the skewfield of pseudodifferential operators over $\mc V$.
For $H(\partial)=\sum_{n=-\infty}^Nh_n\partial^n\in\mc V((\partial^{-1}))$,
we denote by $H(\lambda)$ its symbol:
$H(\lambda)=\sum_{n=-\infty}^Nh_n\lambda^n\in\mc V((\lambda^{-1}))$
(where $\lambda$ is a variable commuting with $\mc V$).
A closed formula for the associative product in $\mc V((\partial^{-1}))$
in terms of the corresponding symbols is the following:
\begin{equation}\label{symbols}
(A\circ B)(\lambda)=A(\lambda+\partial)B(\lambda)\,.
\end{equation}
Here and further we always expand an expression as $(\lambda+\partial)^n,\,n\in\mb Z$,
in non-negative powers of $\partial$,
so that,
for $A(\partial)=\sum_{m=-\infty}^Na_m\partial^m$ 
and $B(\partial)=\sum_{n=-\infty}^Nb_n\partial^n$,
the RHS of \eqref{symbols} means 
$\sum_{m,n=-\infty}^N\sum_{k=0}^\infty \binom{m}{k} a_mb_n^{(k)} \lambda^{m+n-k}$.

An invertible element of $\Mat_{\ell\times\ell}\mc V((\partial^{-1}))$ is called
non-degenerate.
In particular, a matrix differential operator $A\in\Mat_{\ell\times\ell}\mc V[\partial]$ 
is called non-degenerate if it is invertible in $\Mat_{\ell\times\ell}\mc V((\partial^{-1}))$.


Denote by $\mc V(\partial)$ the minimal skew-subfield 
of the skew-field of pseudodifferential operators $\mc V((\partial^{-1}))$
containing the subalgebra $\mc V[\partial]$ of differential operators over $\mc V$.
Elements of $\mc V(\partial)$ are called \emph{rational} pseudodifferential operators.
Elements of $\Mat_{\ell\times\ell}\mc V(\partial)$ are called 
\emph{rational matrix pseudodifferential operators}.

We often manipulate with rational pseudodifferential operators without expanding them
as Laurent series in $\partial^{-1}$.
Then their symbols are computed using formula \eqref{symbols}.

\begin{theorem}[\cite{CDSK12b}]\label{20120720:thm1}
Let $H(\partial)\in\Mat_{\ell\times\ell}\mc V(\partial)$. Then
\begin{enumerate}[(a)]
\item
$H$ has a fractional decomposition $H=AB^{-1}$,
with $A,B\in\Mat_{\ell\times\ell}\mc V[\partial]$,
and $B$ non-degenerate,
which is minimal in the sense that $\ker(A)$ and $\ker(B)$
have zero intersection in any differential field extension of $\mc V$.
\item
Any other fractional decomposition $H=\tilde A\tilde B^{-1}$ of $H$ can be obtained
by multiplying $A$ and $B$ on the right by the same element 
of  $\Mat_{\ell\times\ell}\mc V[\partial]$.
\end{enumerate}
\end{theorem}


Let $H\in\Mat_{\ell\times\ell}\mc V(\partial)$ 
be an $\ell\times\ell$ rational matrix pseudodifferential operator.
It is called a \emph{Poisson structure} on $\mc V$
if it is skewadjoint and it satisfies the Jaocobi identity.
In order to explain what the latter condition means, we introduce
the non-local $\lambda$-bracket 
$\{\cdot\,_\lambda\,\cdot\}:\,\mc V\otimes\mc V\to\mc V((\lambda^{-1}))$
corresponding to $H$,
given by the following \emph{Master Formula} (cf. \cite{DSK06})
\begin{equation}\label{20110922:eq1}
\{f_\lambda g\}
=
\sum_{\substack{i,j\in I \\ m,n\in\mb Z_+}} 
\frac{\partial g}{\partial u_j^{(n)}}
(\lambda+\partial)^n
H_{ji}(\lambda+\partial)
(-\lambda-\partial)^m
\frac{\partial f}{\partial u_i^{(m)}}
\,.
\end{equation}
In particular
$\{{u_i}_\lambda{u_j}\}_H=H_{ji}(\lambda)$,
the symbol of the operator $H$.
%
%
This $\lambda$-bracket satisfies the following properties:
\begin{enumerate}[(i)]
\item
sesquilinearty:
$\{\partial f_\lambda g\}=-\lambda\{f_\lambda g\}$,
$\{f_\lambda\partial g\}=(\lambda+\partial)\{f_\lambda g\}$;
\item
left Leibniz rule:
$\{f_\lambda gh\}=\{f_\lambda g\}h+\{f_\lambda h\}g$;
\item
right Leibniz rule:
$\{fg_\lambda h\}={\{f_{\lambda+\partial} h\}}_\to g+{\{g_{\lambda+\partial} h\}}_\to f$,
where the arrow means that $\partial$ should be moved to the right.
\end{enumerate}
In fact, the above properties, rather than the Master Formula \eqref{20110922:eq1}, 
are used in most computations of the $\lambda$-brackets.

For a matrix pseudodifferential operator 
$H(\partial)\in\Mat_{\ell\times\ell}\mc V((\partial^{-1}))$ and elements  $f,g,h\in\mc V$, we have
$\{f_\lambda\{g_\mu h\}\}\in\mc V((\lambda^{-1}))((\mu^{-1}))$,
$\{g_\mu\{f_\lambda h\}\}\in\mc V((\mu^{-1}))((\lambda^{-1}))$,
and $\{\{f_\lambda g\}_{\lambda+\mu} h\}\in\mc V(((\lambda+\mu)^{-1}))((\lambda^{-1}))$.
On the other hand, if $H(\partial)\in\Mat_{\ell\times\ell}\mc V(\partial)$ 
is a rational matrix pseudodifferential operator,
all these three elements lie in the same subspace 
$\mc V_{\lambda,\mu}=\mc V[[\lambda^{-1},\mu^{-1},(\lambda+\mu)^{-1}]][\lambda,\mu]$, 
\cite[Cor.3.11]{DSK12}.
Hence, it makes sense to write the following \emph{Jacobi identity}:
$$
\{f_\lambda\{g_\mu h\}\}-\{g_\mu\{f_\lambda h\}\}=\{\{f_\lambda g\}_{\lambda+\mu} h\}\,.
$$
Note that the Jacobi identity holds for every $f,g,h\in\mc V$
if and only if it holds for every triple of generators $u_i,u_j,u_k$, \cite{DSK12}.
\begin{definition}\label{def:hamstr}
A \emph{non-local Poisson structure} on $\mc V$
is a rational matrix pseudodifferential operator $H(\partial)\in\Mat_{\ell\times\ell}\mc V(\partial)$,
which is skewadjoint,
and satisfies the Jacobi identity for any triple of generators:
$$
\{{u_i}_\lambda\{{u_j}_\mu {u_k}\}\}
-\{{u_j}_\mu\{{u_i}_\lambda {u_k}\}\}
=\{\{{u_i}_\lambda {u_j}\}_{\lambda+\mu} {u_k}\}
\,\,,\,\,\,\,i,j,k\in I\,.
$$
Two non-local Poisson structures $H$ and $K$ on $\mc V$
are \emph{compatible} if any their linear combination $\alpha H+\beta K$,
for $\alpha,\beta\in\mc C$ is a non-local Poisson structure.
\end{definition}
Recall that, if $H(\partial)\in\Mat_{\ell\times\ell}\mc V(\partial)$
is a non-local Poisson structure on $\mc V$,
then the corresponding non-local $\lambda$-bracket,
given by the Master Formula \eqref{20110922:eq1},
defines a structure of a \emph{non-local Poisson vertex algebra} on $\mc V$, \cite{DSK12}.

\subsection{Hamiltonian equations and integrability}
\label{sec:2.2}

Let $H\in\Mat_{\ell\times\ell}\mc V(\partial)$ be a non-local Poisson structure 
on the field of differential functions $\mc V$.
\begin{definition}\label{20120124:def}
We say that elements $\tint f\in\mc V/\partial\mc V$ and $P\in\mc V^\ell$
are $H$-\emph{associated},
and we use the notation $\tint f\ass{H} P$ (or $P\ass{H}\tint f$),
if there exist a fractional decomposition $H=AB^{-1}$
and an element $F\in\mc V^\ell$ such that
$$
\frac{\delta f}{\delta u}=B(\partial)F
\,\,,\,\,\,\,P=A(\partial)F\,.
$$
In this case,
we say that $\tint f$ is a \emph{Hamiltonian functional} for $H$,
and $P$ is a \emph{Hamiltonian vector field} for $H$.
\end{definition}
\begin{remark}\label{20120201:rem3}
Since, by Theorem \ref{20120720:thm1},
any fractional decomposition for $H$ is obtained from a minimal one
by multiplying both $A$ and $B$ on the right by the same matrix differential operator,
it is clear that in the notion of $H$-association it suffices to use a fixed minimal fractional
decomposition $H=A_0B_0^{-1}$.
In particular, 
the space of all Hamiltonian functionals for $H$ is
\begin{equation}\label{eq:hamfun}
\mc F(H)=
\Big(\frac{\delta}{\delta u}\Big)^{-1}(\im B_0)\,\subset\mc V/\partial\mc V\,,
\end{equation}
and the space of all Hamiltonian vector fields for $H$ is
\begin{equation}\label{eq:hamvec}
\mc H(H)=
A_0\Big(B_0^{-1}\Big(\im\frac{\delta}{\delta u}\Big)\Big)\,\subset\mc V^\ell\,.
\end{equation}
\end{remark}
%
%
\begin{lemma}\label{20120907:lem1}
\begin{enumerate}[(a)]
\item
If $\tint f\ass{H}P$ and $\tint g\ass{H}Q$, then $\tint(a f+b g)\ass{H}(aP+bQ)$ for every $a,b\in\mc C$.
In particular, $\mc F(H)$ and $\mc H(H)$ are vector spaces over $\mc C$.
\item
If $\tint f\ass{H}P$, then
$\Big\{\tint g\in\mc F(H)\,\Big|\,\tint g\ass{H}P\Big\}=\tint f+\mc F_0(H)$,
where 
\begin{equation}\label{20120908:eq1}
\mc F_0(H)=\Big\{\tint g\in\mc F(H)\,\Big|\,\tint g\ass{H}0\Big\}\,.
\end{equation}
\item
If $\tint f\ass{H}P$, then
$\Big\{Q\in\mc H(H)\,\Big|\,\tint f\ass{H}Q\Big\}=P+\mc H_0(H)$,
where 
\begin{equation}\label{20120908:eq2}
\mc H_0(H)=\Big\{Q\in\mc H(H)\,\Big|\,0\ass{H}Q\Big\}\,.
\end{equation}
\end{enumerate}
\end{lemma}
\begin{proof}
Obvious.
\end{proof}
\begin{proposition}[{\cite[Sec.7.1]{DSK12}}]\label{prop:lieham}
\begin{enumerate}[(a)]
\item
The space $\mc V^\ell$ has a Lie algebra structure with Lie bracket 
$$
[P,Q]_i
=
\sum_{j\in I,n\in\mb Z_+}
\frac{\partial Q_i}{\partial u_j^{(n)}}\partial^nP_j
-\frac{\partial P_i}{\partial u_j^{(n)}}\partial^nQ_j
\,,
$$
and $\mc H(H)\subset\mc V^\ell$ is its subalgebra.
\item
We have a representation $\phi$ of the Lie algebra $\mc V^\ell$ 
on the space $\mc V/\partial\mc V$ given by
$$
\phi(P)\big(\tint h\big)=\int P\cdot\frac{\delta h}{\delta u}\,,
$$
and the subspace $\mc F(H)\subset\mc V/\partial\mc V$
is preserved by the action of the Lie subalgebra $\mc H(H)\subset\mc V^\ell$.
\item
If $\tint h\ass{H}P$ and $\tint h\ass{H}Q$ for some $\tint h\in\mc F(H)$,
then the action of $P,Q\in\mc H(H)$ on $\mc F(H)$ is the same:
$$
\int P\cdot\frac{\delta g}{\delta u}
=\int Q\cdot\frac{\delta g}{\delta u}
\,\,\text{ for all } \tint g\in\mc F(H)\,.
$$
\item
The space of Hamiltonian functionals $\mc F(H)$
is a Lie algebra with Lie bracket 
\begin{equation}\label{20120124:eq4}
\{\tint f,\tint g\}_H
=
\int P\cdot\frac{\delta g}{\delta u}
\quad
\Big(
=\int \frac{\delta g}{\delta u}\cdot A(\partial) B^{-1}(\partial) \frac{\delta f}{\delta u}
\,\,\Big)\,,
\end{equation}
where $P\in\mc H(H)$ is such that $\tint f\ass{H}P$.
\item
The Lie algebra action of $\mc H(H)$ on $\mc F(H)$ is by derivations
of the Lie bracket \eqref{20120124:eq4}.
\item
The subspace 
$$
\mc A(H)=\Big\{(\tint f,P)\in\mc F(H)\times\mc H(H)\,\Big|\,\tint f\ass{H}P\Big\}\,
$$
is a subalgebra of the direct product of Lie algebras $\mc F(H)\times\mc H(H)$.
\end{enumerate}
\end{proposition}

A \emph{Hamiltonian equation} associated to the Poisson structure $H$
and to the Hamiltonian functional $\tint h\in\mc F(H)$,
with an associated Hamiltonian vector field $P\in\mc H(H)$,
is, by definition, the following evolution equation on the variables $u=\big(u_i\big)_{i\in I}$:
\begin{equation}\label{20120124:eq5}
\frac{du}{dt}
=P\,.
\end{equation}
By the chain rule, any element $f\in\mc V$ evolves according to the equation
$$
\frac{df}{dt}=\sum_{i\in I}\sum_{n\in\mb Z_+}(\partial^nP_i)\frac{\partial f}{\partial u_i^{(n)}}\,,
$$
and, integrating by parts,
a local functional $\tint f\in\mc V/\partial\mc V$
evolves according to
$$
\frac{d\tint f}{dt}=\int P\cdot\frac{\delta f}{\delta u}\,.
$$
An \emph{integral of motion} for the Hamiltonian equation \eqref{20120124:eq5}
is a Hamiltonian functional $\tint f\in\mc F(H)$ such that
$$
\frac{d\tint f}{dt}=\{\tint h,\tint f\}_H=0\,,
$$
i.e. $\tint f$ lies in the centralizer of $\tint h$ in the Lie algebra $\mc F(H)$.
The Hamiltonian equation \eqref{20120124:eq5} is said to be \emph{integrable}
if there is an infinite sequence of pairs $(\tint h_n,P_n)\in\mc F(H)\times\mc H(H),\,n\geq0$,
such that $\tint h_0=\tint h$, $P_0=P$,
we have $\tint h_n\ass{H}P_n$ for every $n\in\mb Z_+$,
the sequence $\{\tint h_n\}_{n\in\mb Z_+}$ spans 
an infinite dimensional abelian subalgebra of the Lie algebra $\mc F(H)$,
and the sequence $\{P_n\}_{n\in\mb Z_+}$ spans an infinite dimensional 
abelian subalgebra of the Lie algebra $\mc H(H)$.
Equivalently,  equation \eqref{20120124:eq5} is integrable if 
there exists an abelian subalgebra $\mf h$
of the Lie algebra $\mc A(H)$ defined in Proposition \ref{prop:lieham}(f),
such that both its canonical projections $\pi_1(\mf h)$ and $\pi_2(\mf h)$ 
in $\mc F(H)$ and $\mc H(H)$ respectively are infinite dimensional, 
and $\tint h\in\pi_1(\mf h)$.
In this case, we have an \emph{integrable hierarchy} of Hamiltonian equations
$$
\frac{du}{dt_n} = P_n\,,\,\,n\in\mb Z_+\,.
$$


Sometimes a non-local Poisson structure $H$ admits a nice decomposition
of the form
\begin{equation}\label{20121018:eq1}
H=A_1B_1^{-1}A_2B_2^{-1}\dots A_nB_n^{-1}\,,
\end{equation}
where $A_1,\dots,A_n,B_1,\dots,B_n$ are $\ell\times\ell$ matrix differential operators
and the $B_i$'s are-non degenerate.
In this case, to check that $\tint f\in\mc V/\partial\mc V$ and $P\in\mc V^\ell$
are $H$-associated one can use the following result.
\begin{proposition}\label{20121018:prop1}
Let $H\in\Mat_{\ell\time\ell}\mc V(\partial)$ be a non-local Poisson structure
of the form \eqref{20121018:eq1}.
Then we have
$\tint f\ass{H}P$, provided that there exist $F_1,\dots,F_n\in\mc V^\ell$
such that
\begin{equation}\label{20121018:eq2}
P=A_1F_1
\,\,,\,\,\,\,
B_1F_1=A_2F_2
\,\,,\,\,\dots\,\,,\,\,
B_{n-1}F_{n-1}=A_nF_n
\,\,,\,\,\,\,
B_nF_n=\frac{\delta f}{\delta u}\,.
\end{equation}
\end{proposition}
For the proof we use the following:
\begin{lemma}[\cite{CDSK12c}]\label{20121018:lem1}
Let $A$ and $B$ be $\ell\times\ell$ matrix differential operators, with $B$ non-degenerate,
and let $R=A\tilde{B}=B\tilde{A}$ be their right least common multiple.
If $F,G\in\mc V^\ell$ solves $BF=AG$,
then there exists $\tilde{F}\in\mc V^\ell$ such that $F=\tilde{B} \tilde{F},\,G=\tilde{A} \tilde{F}$.
\end{lemma}
\begin{proof}[Proof of Proposition \ref{20121018:prop1}]
For $n=1$ the above condition \eqref{20121018:eq2}
reduces to the Definition \ref{20120124:def} of $H$-association,
so there is nothing to prove.
We prove the statement by induction on $n\geq2$.
Let $R=A_n\tilde{B}_{n-1}=B_{n-1}\tilde{A}_n$
be the least right common multiple of $A_n$ and $B_{n-1}$.
Then, $B_{n-1}^{-1}A_n=\tilde{A}_n {\tilde{B}_{n-1}}^{-1}$,
and the decomposition \eqref{20121018:eq1} gives
$$
H=A_1B_1^{-1}\dots A_{n-2}B_{n-2}^{-1}(A_{n-1}\tilde{A}_n)(B_n\tilde{B}_{n-1})^{-1}\,.
$$
By assumption, we have $B_{n-1}F_{n-1}=A_nF_n$.
Hence, by Lemma \ref{20121018:lem1}, there exists $\tilde{F}_{n-1}\in\mc V^\ell$
such that $F_{n-1}=\tilde{A}_n \tilde{F}_{n-1}$ and $F_n=\tilde{B}_{n-1} \tilde{F}_{n-1}$.
It follows that
$A_{n-1}F_{n-1}=(A_{n-1}\tilde{A}_n)\tilde{F}_{n-1}$,
and $(B_n\tilde{B}_{n-1})\tilde{F}_{n-1}=B_nF_n=\frac{\delta f}{\delta u}$.
Therefore, the assumption \eqref{20121018:eq2} implies
$P=A_1F_1$,
$B_1F_1=A_2F_2$, ... ,
$B_{n-2}F_{n-2}=(A_{n-1}\tilde{A}_n) \tilde{F}_{n-1}$,
$(B_n\tilde{B}_{n-1})\tilde{F}_{n-1}=\frac{\delta f}{\delta u}$.
By the inductive assumption, we conclude that $\tint f\ass{H}P$.
\end{proof}

\subsection{The Lenard-Magri scheme of integrability}
\label{sec:2.3}

The so-called \emph{Lenard-Magri scheme of integrability}
provides an algorithm to construct hierarchies of integrable
bi-Hamiltonian equations.
Sufficient conditions for the applicability of this scheme
are given in the following theorem,
which summarizes the results in \cite[Sec.7.2]{DSK12}.
\begin{theorem}\label{th:lmscheme}
Let $H$ and $K\in\Mat_{\ell\times\ell}\mc V(\partial)$
be compatible non-local Poisson structures 
on a field of differential functions $\mc V$ (in $\ell$ variables $u_1,\dots,u_\ell$),
and assume that $K$ is non-degenerate.
Let, for some $N\in\mb Z_+$, $\{\tint h_n\}_{n=0}^N\subset\mc F(H)$
and $\{P_n\}_{n=0}^N\subset\mc H(H)$ be such that
\begin{equation}\label{20120907:eq2}
\tint 0\ass{H}P_0\ass{K}\tint h_0\ass{H} P_1\ass{K}\tint h_1\ass{H}\dots\ass{H} P_N\ass{K}\tint h_N\,,
\end{equation}
and such that the following orthogonality conditions
hold for some fractional decompositions $H=AB^{-1}$ and $K=CD^{-1}$:
\begin{equation}\label{20120621:eq2}
\big(\Span{}_{\mc C}\big\{\frac{\delta h_n}{\delta u}\big\}_{n=0}^N\big)^\perp\subset\im(C)
\,\,,\,\,\,\,
\big(\Span{}_{\mc C}\big\{P_n\big\}_{n=0}^N\big)^\perp \subset\im(B)\,,
\end{equation}
where the orthogonal complement is with respect to the pairing
$\mc V^{\ell}\times\mc V^\ell\to\mc V/\partial\mc V,\,(\xi,P)\mapsto\tint \xi\cdot P$.
Then the given finite sequences can be extended 
to infinite sequences $\{\tint h_n\}_{n=0}^\infty$, 
$\{P_n\}_{n=0}^\infty$, 
such that 
\begin{equation}\label{20120907:eq3}
\tint h_{n-1}\ass{H} P_n\ass{K}\tint h_n
\,\,\text{ for all } n\geq 1\,,
\end{equation}
in some field of differential functions extension $\tilde{\mc V}$ of $\mc V$.
(In fact, the $P_n$'s and $\frac{\delta h_n}{\delta u}$'s can be taken in $\mc V^\ell$.)

The elements $\tint h_n,\,n\in\mb Z_+$, are in involution
with respect to both Lie brackets \eqref{20120124:eq4} for $H$ and $K$
(in the Lie algebras $\tilde{\mc F}(H)$ and $\tilde{\mc F}(K)$, 
subspaces of $\tilde{\mc V}/\partial\tilde{\mc V}$):
\begin{equation}\label{involution}
\{\tint h_m,\tint h_n\}_H=\{\tint h_m,\tint h_n\}_K=0
\,\,,\,\,\,\,\text{ for all } m,n\in\mb Z_+\,,
\end{equation}
and,
\begin{equation}\label{compatible}
[P_m,P_n]\in\ker(B^*)\cap\ker(D^*)
\,\,\text{ for all } m,n\in\mb Z_+\,.
\end{equation}

In particular, if $\ker(B^*)\cap\ker(D^*)=0$,
then the evolution equations
\begin{equation}\label{20120721:eq1}
\frac{du}{dt_n}=P_n\,\,,\,\,\,\,\,n\in\mb Z_+\,,
\end{equation}
form a compatible hierarchy,
and the $\tint h_n$'s
are its integrals of motion;
therefore, if the subspaces 
$\Span_{\tilde{\mc C}}\{\tint h_n\}_{n\in\mb Z_+}\subset\tilde{\mc V}/\partial\tilde{\mc V}$
and $\Span_{\mc C}\{P_n\}_{n\in\mb Z_+}\subset\mc V^\ell$
are infinite dimensional,
then \eqref{20120721:eq1} is an integrable hierarchy of Hamiltonian equations.
\end{theorem}

\begin{remark}\label{rem:nonminimal}
In \cite[Sec.7.2]{DSK12} the above result is stated under the assumption that 
the algebra of differential functions $\mc V$ is normal.
The present formulation is basically equivalent due to the fact 
that any algebra of differential functions can be extended to a normal one.
Moreover, in the same paper we assume that the fractional decompositions
$H=AB^{-1}$ and $K=CD^{-1}$ are minimal.
However, it is clear that, if Theorem \ref{th:lmscheme} holds 
for minimal fractional decompositions $H=A_0B_0^{-1},\,K=C_0D_0^{-1}$,
then it automatically holds for arbitrary fractional decompositions
$H=AB^{-1},\,K=CD^{-1}$.
Indeed, thanks to Theorem \ref{20120720:thm1},
$\im C\subset\im C_0,\,\im B\subset\im B_0,\,
\ker B_0^*\subset\ker B^*,\,\ker D_0^*\subset\ker D^*$.
\end{remark}

\begin{remark}\label{20120906:rem1}
The condition that $\ker(B^*)\cap\ker(D^*)=0$ is actually not needed 
for the proof of integrability of each given equation $\frac{du}{dt}=P_n$.
Indeed, since $\ker(B^*)\cap\ker(D^*)$ is finite dimensional over $\mc C$,
then the space $\Span_{\mc C}\{P_n\}_{n\in\mb Z_+}\subset\mc V^\ell$
contains an infinite dimensional abelian subalgebra containing any given element of it.
This is true due to the following result.
\end{remark}
\begin{lemma}\label{20120906:lem1}
Let $U$ be an infinite dimensional subspace of a Lie algebra
such that $[U,U]$ is finite dimensional.
Then any element of $U$ is contained in an infinite dimensional abelian subalgebra of $U$.
\end{lemma}
\begin{proof}
Let $a_1$ be a non-zero element of $U$.
The centralizer $C_1$ of $a_1$ in $U$
is the kernel of the map $\ad a:\,U\to[U,U]$,
hence, it has finite codimension in $U$.
Next, let $a_2$ be an element of $C_1$ linearly independent of $a_1$,
and let $C_2$ be its centralizer in $C_1$.
By the same argument, $C_2$ has finite codimension in $C_1$.
In this fashion we construct an infinite sequence of linearly independent
commuting elements of $U$.
\end{proof}

According to Theorem \ref{th:lmscheme}, 
when we apply the Lenard-Magri scheme of integrability
we get infinite sequences 
$\{\tint h_n\}_{n=0}^\infty\subset\mc V/\partial\mc V$, $\{P_n\}_{n=0}^\infty\subset\mc V^\ell$,
satisfying the relations \eqref{20120907:eq3}.
At this point, in order to prove integrability,
we still need to prove that the spaces $\Span_{\mc C}\{\tint h_n\}_{n=0}^\infty$
and $\Span_{\mc C}\{P_n\}_{n=0}^\infty$ are infinite dimensional.
(In fact, it is easy to see that it suffices to check the latter condition).
For this one usually uses differential order considerations, based on the following result.
\begin{proposition}\label{20120910:prop}
Suppose that the rational $\ell\times\ell$-matrix pseudodifferential operators 
$H=AB^{-1},\,K=CD^{-1}\,\in\Mat_{\ell\times\ell}\mc V(\partial)$, are in their minimal fractional decompositions.
Assume also that the matrix differential operators $A,B,C,D$ have non-degenerate
leading coefficients.
Let $P,Q\in\mc V^\ell$ and $\tint h\in\mc V/\partial\mc V$ satisfy the following relations
\begin{equation}\label{20120910:eq3}
P\ass{K}\tint h\ass{H}Q\,,
\end{equation}
and assume that
\begin{equation}\label{20120911:eq1}
\dord(P)>\max\!\big\{\!\dord(A)-|H|+|K|,\dord(B)+|K|,\dord(C),\dord(D)+|K|\big\}.
\end{equation}
(Here we use the notation introduced in \eqref{20120910:eq1} and \eqref{20120910:eq2}).
Then
$$
\dord\big(\frac{\delta h}{\delta u}\big)=\dord(P)-|K|
\,\,\,\,\text{ and }\,\,
\dord(Q)=\dord(P)+|H|-|K|\,.
$$
\end{proposition}
\begin{proof}
By definition, the relations \eqref{20120910:eq3} amount to the existence of elements
$F,G\in\mc V^\ell$ such that
$$
CG=P
\,\,,\,\,\,\,
DG=\frac{\delta h}{\delta u}
\,\,,\,\,\,\,
BF=\frac{\delta h}{\delta u}
\,\,,\,\,\,\,
AF=Q
\,.
$$
Since $C$ has non-degenerate leading coefficient 
and $\dord(P)=\dord(CG)>\dord(C)$, 
we get by Lemma \ref{20120910:lem1}(c) that 
$$
\dord(G)=\dord(CG)-|C|=\dord(P)-|C|\,.
$$
Next, since by assumption $D$ has non-degenerate leading coefficient 
and $\dord(G)+|D|=\dord(P)-|C|+|D|=\dord(P)-|K|>\dord(D)$, 
we get by Lemma \ref{20120910:lem1}(b) that 
$$
\begin{array}{c}
\displaystyle{
\dord(\frac{\delta h}{\delta u})=\dord(DG)=\dord(G)+|D|=\dord(P)-|C|+|D|
} \\
\displaystyle{
=\dord(P)-|K|\,.
}
\end{array}
$$
Similarly, since, by assumption, $B$ has non-degenerate leading coefficient 
and $\dord(BF)=\dord(\frac{\delta h}{\delta u})=\dord(P)-|K|>\dord(B)$,
we get by Lemma \ref{20120910:lem1}(c) that 
$$
\dord(F)=\dord(BF)-|B|=\dord(\frac{\delta h}{\delta u})-|B|=\dord(P)-|K|-|B|\,.
$$
Finally, since, by assumption, $A$ has non-degenerate leading coefficient 
and $\dord(F)+|A|=\dord(P)-|K|=|B|+|A|=\dord(P)-|K|+|H|>\dord(A)$, 
we get by Lemma \ref{20120910:lem1}(b) that 
$$
\begin{array}{l}
\displaystyle{
\dord(Q)=\dord(AF)=\dord(F)+|A|=\dord(P)-|K|-|B|+|A|
}\\
\displaystyle{
=\dord(P)-|K|+|H|\,.
}
\end{array}
$$
\end{proof}

Recall from \cite[Theorem 5.1]{DSK12} that, 
given two compatible non-local Poisson structures $H$ and $K$, 
with $\det(K)\neq0$ ,
we get an infinite family of compatible non-local Poisson structures
given by $H^{[0]}=K$, and $H^{[s]}=(H\circ K^{-1})^{s-1}\circ H$, for $s\geq1$.
\begin{proposition}\label{20121018:rem}
Suppose that the sequences
$\{\tint h_n\}_{n\in\mb Z_+}\subset\mc F(H)\cap\mc F(K)$,
$\{P_n\}_{n\in\mb Z_+}\subset\mc H(H)\cap\mc H(K)$
satisfy the Lenard-Magri recursive relations
$\tint h_{n-1}\ass{H} P_n\ass{K}\tint h_n$ for every $n\in\mb Z_+$ 
(we let $\tint h_{-1}=\tint 0$).
Then, we have
$\tint h_n\ass{H^{[s]}} P_{n+s}$ for every $n,s\in\mb Z_+$.
Consequently, all the evolution equations $\frac{du}{dt_n}=P_n$
are Hamiltonian with respect to all the non-local Poisson structures $H^{[s]},\,s\in\mb Z_+$.
\end{proposition}
\begin{proof}
By assumption we have
$\tint h_n\ass{H}P_{n+1}\ass{K}\dots\ass{K}\tint h_{n+s-1}\ass{H}P_{n+s}$,
which means, by definition, that there exist $F_n,F_{n+1},\dots,F_{n+2s-2}\in\mc V^\ell$
such that
$$
\begin{array}{l}
\displaystyle{
\frac{\delta h_n}{\delta u}=BF_n
\,\,,\,\,\,\,
AF_{n+2i-2}=P_{n+i}=CF_{n+2i-1}
\,\,\text{ for } 1\leq i\leq s-1
\,,} \\
\displaystyle{
DF_{n+2i-1}=\frac{\delta h_{n+i}}{\delta u}=BF_{n+2i}
\,\,\text{ for } 1\leq i\leq s-1
\,\,,\,\,\,\,
AF_{n+2s-2}=P_{n+s}\,.
}
\end{array}
$$
On the other hand, the non-local Poisson structure $H^{[s]}$
admits the decomposition
$$
H^{[s]}=AB^{-1}DC^{-1}\dots AB^{-1}DC^{-1}AB^{-1}\,,
$$
i.e. it is of the form \eqref{20121018:eq1}
with $n=2s+1$,
and with $A_i=A$ and $B_i=B$ for odd $i$,
and $A_i=D$ and $B_i=C$ for even $i$.
The claim follows from Proposition \ref{20121018:prop1}.
\end{proof}

\subsection{Lenard-Magri schemes of S-type and C-type}
\label{sec:3.2}

Consider a Lenard-Magri scheme as in \eqref{20120907:eq2}.
We say that it is \emph{integrable} if it can be extended indefinitely,
and the linear spans of $\{\tint h_n\}_{n\in\mb Z_+}$ and $\{P_n\}_{n\in\mb Z_+}$
are infinite dimensional;
in this case, the corresponding hierarchy of evolution Hamiltonian equations
$\frac{du}{dt_n}=P_n,\,n\in\mb Z_+$ is integrable.
We say that the Lenard-Magri scheme \eqref{20120907:eq2}
is \emph{finite} if it can be extended indefinitely,
but in in any such infinite extension 
the linear span of $\{\tint h_n\}_{n\in\mb Z_+}$ or of $\{P_n\}_{n\in\mb Z_+}$
is finite dimensional.
Moreover, we say that the Lenard-Magri scheme \eqref{20120907:eq2}
is \emph{blocked} if it cannot be extended indefinitely,
namely, for some $n$, 
there is no $\tint h_n$ such that $P_n\ass{K}\tint h_n$,
or there is no $P_{n+1}$ such that $\tint h_n\ass{H}P_{n+1}$.

For an integrable Lenard-Magri scheme \eqref{20120907:eq2},
we say that it is of \emph{S-type} the differential orders of the elements $P_n$ can grow to infinity,
and it is of \emph{C-type} if the differential orders of the $P_n$'s are necessarily bounded.
It is easy to see that for an integrable Lenard-Magri scheme of S-type 
the order of the pseudodifferential $H$ should be greater than the order 
of the pseudodifferential $K$.
Indeed, since we have $P_n\ass{K}\tint h_n\ass{H}P_{n+1}$,
if $P_n$ has differential order large enough,
then, by Proposition \ref{20120910:prop}, 
$\dord(P_{n+1})=\dord(P_n)+\ord(H)-\ord(K)$.

\begin{remark}
This terminology in inspired by the terminology of Calogero,
who calls an integrable hierarchy of ``S-type'' if the differential orders 
of the canonical conserved densities are unbounded,
and of ``C-type'' otherwise (see \cite{MSS90,MS12}).
Note that, though, these two terminologies are closed but do not coincide. For example,
the linear equation $\frac{du}{dt}=u'''$ is C-integrable in Calogero's terminology,
but the corresponding Lenard-Magri scheme, considered for example in \cite{BDSK09},
is integrable of S-type.
\end{remark}

\section{Liouville type integrable systems}
\label{sec:3}

In this section $\mc V$ is a field of differential functions in $u$,
and we assume that $\mc V$ contains all the functions that we encounter
in our computations.

Recall from \cite[Example 4.6]{DSK12} that we have the following triple
of compatible non-local Poisson structures:
$$
L_1=\partial\,,\,\, 
L_2=\partial^{-1}\,,\,\, 
L_3=u'\partial^{-1}\circ u'\,. 
$$
Given two non-local Poisson structures $H$ and $K$ of the form
\begin{equation}\label{20121006:eq2}
H=a_1L_1+a_2L_2+a_3L_3
\,\,,\,\,\,\,
K=b_1L_1+b_2L_2+b_3L_3\,,
\end{equation}
with $a_i,b_i\in\mc C,\,i=1,2,3$,
we want to discuss the integrability of the corresponding Lenard-Magri scheme.

\subsection{Preliminary computations}
\label{sec:3.1}

First, we find a minimal fractional decomposition for the operators $H$ and $K$.
\begin{lemma}\label{lem:frac}
For $X=x_1L_1+x_2L_2+x_3L_3$, with $x_1,x_2,x_3\in\mc C$, we have
\begin{equation}\label{frac-liouv}
X=\Big[x_1\partial^2\circ\frac1{u''}\partial+\frac{x_2+x_3(u')^2}{u''}\partial-x_3u'\Big]
\Big[\partial\circ\frac1{u''}\partial\Big]^{-1}\,.
\end{equation}
The above fractional decomposition is minimal only for $x_2x_3\neq0$.
For $x_2\neq0,\,x_3=0$, 
the minimal fractional decomposition is 
\begin{equation}\label{frac-liouv1}
X=(x_1\partial^2+x_2)\partial^{-1}\,,
\end{equation}
for $x_2=0,\,x_3\neq0$ it is 
\begin{equation}\label{frac-liouv2}
X=\Big(x_1\partial\circ\frac1{u'}\partial+x_3u'\Big)
\Big(\frac1{u'}\partial\Big)^{-1}\,,
\end{equation}
and for $x_2=x_3=0$ it is 
$X=x_1\partial$.
\end{lemma}
\begin{proof}
Straightforward.
\end{proof}

Later we will need the following simple facts concerning the numerators 
of the fractional decompositions for $X$.
\begin{lemma}\label{20120909:lem1}
\begin{enumerate}[(a)]
\item
For $x_1,x_2,x_3\in\mc C$, $x_1\neq0$, consider the equation
\begin{equation}\label{20120909:eq1}
\Big(x_1\partial^2\circ\frac1{u''}\partial+\frac{x_2+x_3(u')^2}{u''}\partial-x_3u'\Big)F=f\,,
\end{equation}
in the variables $F\in\mc V$ and $f\in\mc F$.
If $x_3\neq0$,
then all the solutions of equation \eqref{20120909:eq1} are
$$
F=\alpha u'
\,\,,\,\,\,\,
f=x_2\alpha
\,\,\text{ for some }\,\alpha\in\mc C\,,
$$
while if $x_3=0$, 
then all the solutions of equation \eqref{20120909:eq1} are
$$
F=\alpha u'+\beta (xu'-u)+\gamma
\,\,,\,\,\,\,
f=x_2(\alpha+\beta x)
\,\,\text{ for some }\,\alpha,\beta,\gamma\in\mc C\,.
$$
\item
For $x_1,x_2\in\mc C$, $x_1\neq0$, 
an element $F\in\mc V$ satisfies
\begin{equation}\label{20120909:eq2}
(x_1\partial^2+x_2)F\in\mc F
\end{equation}
if and only if $F\in\mc F$.
\item
For $x_1,x_3\in\mc C$, $x_1\neq0$, 
an element $F\in\mc V$ satisfies
\begin{equation}\label{20120909:eq3}
\Big(x_1\partial\circ\frac1{u'}\partial+x_3u'\Big)F\in\mc V_1
\end{equation}
if and only if $F\in\mc V_0$ and $F'=\frac{\partial F}{\partial u}u'$.
\end{enumerate}
\end{lemma}
\begin{proof}
If $n\geq2$ and $F\in\mc V$ solves equation \eqref{20120909:eq1}
and has differential order less than or equal to $n$, then,
using \eqref{eq:0.4}, we have
$$
0=\frac{\partial}{\partial u^{(n+3)}}LHS\eqref{20120909:eq1}
=x_1\frac1{u''}\frac{\partial F}{\partial u^{(n)}}\,,
$$
implying that $\frac{\partial F}{\partial u^{(n)}}=0$.
Hence $F$ must have differential order at most $1$.
Then we have
$$
0=\frac{\partial}{\partial u^{(4)}}LHS\eqref{20120909:eq1}
=x_1\Big(\frac1{u''}\frac{\partial F}{\partial u'}-\frac1{(u'')^2}F'\Big)\,,
$$
so that $F'=\frac{\partial F}{\partial u'}u''$.
But then equation \eqref{20120909:eq1} becomes
\begin{equation}\label{20120730:eq1}
x_1\Big(\frac{\partial F}{\partial u'}\Big)^{''}+(x_2+x_3(u')^2)\frac{\partial F}{\partial u'}-x_3u'F=f\,.
\end{equation}
If $\frac{\partial F}{\partial u'}$ has differential order $n\geq0$,
then applying $\frac{\partial}{\partial u^{(n+2)}}$ to both sides of equation \eqref{20120730:eq1}
we get $\frac{\partial^2 F}{\partial u^{(n)}\partial u'}=0$.
Hence, $\frac{\partial F}{\partial u'}=\varphi$ must be a quasiconstant.
In other words, $F=\varphi u'+f_0$, where $f_0\in\mc V_0$ has differential 
order less than or equal to $0$.
But then the condition $F'=\frac{\partial F}{\partial u'}u''$ becomes
$\varphi' u'+f_0'=0$.
This implies, using \eqref{20120907:eq1},
that $(f_0+\varphi'u)'=\varphi''u\in\partial\mc V\cap\mc V_0=\partial\mc F$.
So, necessarily, $\varphi''=0$.
Hence, $\varphi=\alpha+\beta x$
and $f_0=-\beta u+\gamma$, for some constants $\alpha,\beta,\gamma\in\mc C$.
Putting these results together, we have
$$
F=(\alpha+\beta x) u'-\beta u+\gamma\,,
$$
and plugging back into equation \eqref{20120730:eq1} we get
$$
x_2(\alpha+\beta x)+x_3\beta uu'-x_3\gamma u'=f\,.
$$
Since, by assumption, $f\in\mc F$,
we obtain $\beta=\gamma=0$ if $x_3\neq 0$,
completing the proof of part (a).

For part (b) we just observe that, if $F\in\mc V_n$ for some $n\geq0$
satisfies condition \eqref{20120909:eq2}, then
$$
0=\frac{\partial}{\partial u^{(n+2)}}(x_1F''+x_2F)
=x_1\frac{\partial F}{\partial u^{(n)}}\,.
$$
Hence, $F$ must be a quasiconstant.

Similarly, 
if $F\in\mc V_n$ for some $n\geq1$
satisfies condition \eqref{20120909:eq3}, then
$$
0=\frac{\partial}{\partial u^{(n+2)}}
\Big(x_1\partial\frac{F'}{u'}+x_3u'F\Big)
=\frac{x_1}{u'}\frac{\partial F}{\partial u^{(n)}}\,.
$$
Hence, $F$ must lie in $\mc V_0$.
Furthermore, 
$$
0=\frac{\partial}{\partial u''}
\Big(x_1\partial\frac{F'}{u'}+x_3u'F\Big)
=x_1\Big(-\frac{F'}{(u')^2}+\frac1{u'}
\frac{\partial F}{\partial u}\Big)\,.
$$
Hence, $F$ must be such that $F'=\frac{\partial F}{\partial u}u'$.
\end{proof}

Next, we compute the spaces $\mc F_0(X)$ and $\mc H_0(X)$
defined in \eqref{20120908:eq1} and \eqref{20120908:eq2}.
Here and further we use the following notation: given two constants $x_i,x_j\in\mc C$, we let
\begin{equation}\label{notation}
x_{ij}=\sqrt{-\frac{x_j}{x_i}}\,\in\mc C\,.
\end{equation}
(We assume that the field $\mc C$ contains all such elements).
\begin{lemma}\label{20120908:lem1}
For $X=x_1L_1+x_2L_2+x_3L_3$, we have:
\begin{enumerate}[(a)]
\item
\begin{enumerate}[]
\item
$\mc F_0(X)=\ker\big(\frac\delta{\delta u}\big)$
if $x_1x_2x_3\neq0$;
\item
$\mc F_0(X)=\mc C\tint e^{x_{12}x}u
+\mc C\tint e^{-x_{12}x}u+\ker\big(\frac{\delta}{\delta u}\big)$
if $x_1x_2\neq0,x_3=0$;
\item
$\mc F_0(X)=\mc C\tint e^{x_{13}u}
+\mc C\tint e^{-x_{13}u}+\ker\big(\frac{\delta}{\delta u}\big)$
if $x_1x_3\neq0,x_2=0$;
\item
$\mc F_0(X)=\mc C\tint u+\ker\big(\frac\delta{\delta u}\big)$
if $x_1\neq0$ and $x_2=x_3=0$;
\item
$\mc F_0(X)=\mc C\tint\sqrt{x_2+x_3(u')^2}+\ker\big(\frac\delta{\delta u}\big)$,
if $x_1=0,x_2x_3\neq0$;
\item
$\mc F_0(X)=\ker\big(\frac\delta{\delta u}\big)$ if $x_1=0$ and $x_2=0$ or $x_3=0$.
\end{enumerate}
\item
\begin{enumerate}[]
\item
$\mc H_0(X)=\mc C\oplus \mc Cu'$ if $x_2x_3\neq0$;
\item
$\mc H_0(X)=\mc C$ if $x_2\neq0,x_3=0$;
\item
$\mc H_0(X)=\mc Cu'$ if $x_2=0,x_3\neq0$;
\item
$\mc H_0(X)=0$ if $x_2=x_3=0$.
\end{enumerate}
\end{enumerate}
\end{lemma}
\begin{proof}
First, let us find all elements $P\in\mc H_0(X)$. 
By Remark \ref{20120201:rem3},
if $X=YZ^{-1}$ is a minimal fractional decomposition,
we need to solve the following equations in $F,P\in \mc V$:
\begin{equation}\label{20120908:eq3}
ZF=0\,\,,\,\,\,\,P=YF\,.
\end{equation}
By Lemma \ref{lem:frac},
if $x_2=x_3=0$, then $Y=x_1\partial$ and $Z=1$, so the only solution 
of \eqref{20120908:eq3} is given by $F=0$, $P=0$.
If $x_2\neq0,\,x_3=0$, then $Y=x_1\partial^2+x_2$ and $Z=\partial$, so 
we get $F\in\mc C$ and $P\in\mc C$.
Similarly, if $x_2=0,\,x_3\neq0$, then $Y=x_1\partial\circ\frac1{u'}\partial+x_3u'$ 
and $Z=\frac1{u'}\partial$, so we get $F\in\mc C$ and $P\in\mc Cu'$.
Finally, if $x_2\neq0,\,x_3\neq0$, then 
$Y=x_1\partial^2\circ\frac1{u''}\partial+\frac{x_2+x_3(u')^2}{u''}\partial-x_3u'$
and $Z=\partial\circ\frac1{u''}\partial$.
Hence, the solutions of \eqref{20120908:eq3} are 
$F=\alpha+\beta u'\in\mc C\oplus\mc Cu'$,
and $P=YF=x_2\beta-x_3\alpha u'\in\mc Cu'$.
This proves part  (b).

Next, we find all elements $\tint f\in\mc F_0(X)$,
namely all solutions of the following equations in $F\in \mc V$ and $\tint f\in\mc V/\partial\mc V$:
\begin{equation}\label{20120908:eq4}
YF=0\,\,,\,\,\,\,\frac{\delta f}{\delta u}=ZF\,.
\end{equation}
If $x_1=0,x_2\neq0,x_3=0$, we have $Y=x_2$ is invertible,
and similarly, if $x_1=0,x_2=0,x_3\neq0$, we have $Y=x_3u'$ is invertible too.
In both these cases we thus have $F=0$, and hence 
$\tint f\in\ker\big(\frac\delta{\delta u}\big)$.
If $x_1=0,x_2\neq0,x_3\neq0$, then 
$Y=\frac{x_2+x_3(u')^2}{u''}\partial-x_3u'$ and $Z=\partial\circ\frac1{u''}\partial$.
The equation $YF=0$ has a one-dimensional (over $\mc C$) space of solution, spanned
by $F=\sqrt{x_2+x_3(u')^2}$.
Hence, all elements $\tint f\in\mc F_0(X)$ are obtained solving the equation
$$
\frac{\delta f}{\delta u}=\alpha\partial\circ\frac1{u''}\partial\sqrt{x_2+x_3(u')^2}
=\alpha\Big(\frac{x_3u'}{\sqrt{x_2+x_3(u')^2}}\Big)'\,,
$$
for $\alpha\in\mc C$. 
Its solutions are of the form
$\tint f=-\alpha\sqrt{x_2+x_3(u')^2}+k$,
where $k\in\ker\big(\frac\delta{\delta u}\big)$.
Next, 
if $x_1\neq0,x_2=x_3=0$, then $Y=x_1\partial$ and $Z=1$, so the equations \eqref{20120908:eq4}
give $F\in\mc C$ and $\tint f\in\mc C\tint u+\ker\big(\frac{\delta}{\delta u}\big)$.
If $x_1\neq0,x_2\neq0,\,x_3=0$, then $Y=x_1\partial^2+x_2$ and $Z=\partial$.
In this case, the first equation in \eqref{20120908:eq4} reads
$$
x_1F''+x_2F=0\,.
$$
By Lemma \ref{20120909:lem1}(b), it must be $F\in\mc F$, 
and it is easy to see that 
the space of solutions is two-dimensional over $\mc C$, 
consisting of elements of the form
$$
F=\alpha_+e^{x_{12}x}+\alpha_-e^{-x_{12}x}\,,
$$
with $\alpha_\pm\in\mc C$.
Then, the second equation in \eqref{20120908:eq4} gives
$$
\frac{\delta f}{\delta u}
=\alpha_+x_{12}e^{x_{12}x}
-\alpha_-x_{12}e^{-x_{12}x}\,,
$$
so that 
$\tint f=
\alpha_+x_{12} \tint e^{x_{12}x}u
-\alpha_-x_{12} \tint e^{-x_{12}x}u+k$,
where $k\in\ker\big(\frac{\delta}{\delta u}\big)$.
Similarly, we consider the case $x_1\neq0,x_2=0,x_3\neq0$.
In this case $Y=x_1\partial\circ\frac1{u'}\partial+x_3u'$ and $Z=\frac1{u'}\partial$.
The first equation in \eqref{20120908:eq4} reads
$$
x_1\Big(\frac{F'}{u'}\Big)'+x_3u'F=0\,.
$$
By Lemma \ref{20120909:lem1}(c),
we must have $F\in\mc V_0$ such that $F'=\frac{\partial F}{\partial u}$.
It is easy to see that
the space of solutions is two-dimensional over $\mc C$, and
it consists of elements of the form
$$
F=\alpha_+e^{x_{13}u}+\alpha_-e^{-x_{13}u}\,,
$$
with $\alpha_\pm\in\mc C$.
Then, the second equation in \eqref{20120908:eq4} gives
$$
\frac{\delta f}{\delta u}
=\alpha_+x_{13}e^{x_{13}u}
-\alpha_-x_{13}e^{-x_{13}u}\,,
$$
and its solutions for $\tint f$ are of the form
$\tint f=
\alpha_+\tint e^{x_{13}u}
+\alpha_-\tint e^{-x_{13}u}+k$,
for $k\in\ker\big(\frac{\delta}{\delta u}\big)$.
Finally, we are left to consider the case when
$x_1\neq0,x_2\neq0,\,x_3\neq0$.
In this case
$Y=x_1\partial^2\circ\frac1{u''}\partial+\frac{x_2+x_3(u')^2}{u''}\partial-x_3u'$
and $Z=\partial\circ\frac1{u''}\partial$.
The first equation in \eqref{20120908:eq4} reads
$$
x_1\Big(\frac{F'}{u''}\Big)''+(x_2+x_3(u')^2)\frac{F'}{u''}-x_3u'F=0\,.
$$
By Lemma \ref{20120909:lem1}(a),
the only solution of this equation is $F=0$.
But then the second equation in \eqref{20120908:eq4} 
gives $\tint f\in\ker\big(\frac{\delta}{\delta u}\big)$.
\end{proof}

In the statement of Lemma \ref{20120908:lem1} and further on in this section,
we assume that $\mc V$ contains all the elements which appear in the statement,
namely
$e^{x_{12}x}$, $e^{x_{13}u}$,
and $\sqrt{x_2+x_3(u')^2}$.

Next, for each element $\tint f\in\mc F_0(X)$, we want to find an element $P\in\mc H(X)$
which is $X$-associated to it,
and for each element $P\in\mc H_0(X)$, we want to find an element $\tint f\in\mc H(X)$
which is $X$-associated to it.
Recall, by Lemma \ref{20120907:lem1},
that if $\tint f\ass{X}P$,
then all elements in $\mc H(X)$ which are $X$-associated to $\tint f$
are obtained adding to $P$ an arbitrary element of $\mc H_0(X)$,
and all elements in $\mc F(X)$ which are $X$-associated to $P$
are obtained adding to $\tint f$ an arbitrary element of $\mc F_0(X)$.
\begin{lemma}\label{20120908:lem2}
Let $X=x_1L_1+x_2L_2+x_3L_3$, and let $a_2,a_3,\gamma\in\mc C\backslash\{0\}$. 
We have:
\begin{enumerate}[(i)]
\item
\begin{enumerate}[]
\item
$\nexists P\in\mc H(X)$ such that 
$\tint e^{\gamma x}u\ass{X}P$, if $x_3\neq0$;
\item
$\tint e^{\gamma x}u\ass{X}\frac1\gamma (x_1\gamma^2+x_2)e^{\gamma x}$, if $x_3=0$.
\end{enumerate}
\item
\begin{enumerate}[]
\item
$\nexists P\in\mc H(X)$ such that 
$\tint e^{\gamma u}\ass{X}P$, if $x_2\neq0$;
\item
$\tint e^{\gamma u}\ass{X}(x_1\gamma^2+x_3)e^{\gamma u}u'$, if $x_2=0$.
\end{enumerate}
\item
$\tint u\ass{X}(x_2x+x_3uu')$.
\item
$\tint\sqrt{a_2+a_3(u')^2}\ass{X}
-\Big(x_1\partial^2\circ\frac1{u''}\partial+\frac{x_2+x_3(u')^2}{u''}\partial-x_3u'\Big)\sqrt{a_2+a_3(u')^2}$.
\item
\begin{enumerate}[]
\item 
$\tint 0\ass{X}1$, if $x_2\neq0$;
\item
$\nexists \tint f\in\mc F(X)$ such that 
$\tint f\ass{X}1$, if $x_2=0,x_1x_3\neq0$;
\item
$\tint\frac1{2x_3u'}\ass{X}1$, if $x_1=0,x_2=0,x_3\neq0$;
\item
$\tint\frac{xu}{x_1}\ass{X}1$, if $x_1\neq0,x_2=0,x_3=0$.
\end{enumerate}
\item
\begin{enumerate}[]
\item 
$\tint 0\ass{X}u'$, if $x_3\neq0$;
\item
$\nexists \tint f\in\mc F(X)$ such that 
$\tint f\ass{X}u'$, if $x_3=0,x_1x_2\neq0$;
\item
$\tint\frac{-(u')^2}{2x_2}\ass{X}u'$, if $x_1=0,x_2\neq0,x_3=0$;
\item
$\tint\frac{u^2}{2x_1}\ass{X}u'$, if $x_1\neq0,x_2=0,x_3=0$.
\end{enumerate}
\end{enumerate}
\end{lemma}
\begin{proof}
The condition $\tint e^{\gamma x}u\ass{X}P$
means that,
for some fractional decomposition $X=YZ^{-1}$,
there exists $F\in\mc V$ such that $P=YF$ and $ZF=e^{\gamma x}$.
Let us consider first the case $x_3\neq0$.
In this case, by Lemma \ref{lem:frac} a minimal fractional decomposition for $X$ is
\eqref{frac-liouv} if $x_2\neq0$, and \eqref{frac-liouv2} if $x_2=0$.
In the former case
$Z=\partial\circ\frac1{u''}\partial$, hence the equation 
$ZF=e^{\gamma x}$ reads
$\partial\frac{F'}{u''}=e^{\gamma x}$,
namely
$$
F'=\frac1\gamma e^{\gamma x}u''+\alpha u''\,,
$$
for some $\alpha\in\mc C$.
This equation has no solutions since, integrating by parts,
$\tint \Big(\frac1\gamma e^{\gamma x}u''+\alpha u''\Big)=\gamma\tint e^{\gamma x}u$,
and this is not zero by \eqref{20120907:eq1}.
Similarly, in the case $x_2=0$ we have 
$Z=\frac1{u'}\partial$, hence the equation 
$ZF=e^{\gamma x}$ reads
$$
F'=e^{\gamma x}u'\,,
$$
which has no solutions since, integrating by parts,
$\tint e^{\gamma x}u'=-\gamma\tint e^{\gamma x}u\neq0$.
To conclude the proof of part (i), we consider the case $x_3=0$.
By Lemma \ref{lem:frac} a fractional decomposition for $X$ is
$X=YZ^{-1}$ given by \eqref{frac-liouv1}.
Hence, a solution $F\in\mc V$ to the equation $ZF=e^{\gamma x}$ is $F=\frac1\gamma e^{\gamma x}$,
and in this case we have $P=YF=(x_1\partial^2)\frac1\gamma e^{\gamma x}
=\big(x_1\gamma+\frac{x_2}{\gamma}\big)e^{\gamma x}$.

Next, let us prove part (ii).
The condition $\tint e^{\gamma u}\ass{X}P$
is equivalent to the existence of $F\in\mc V$ such that $P=YF$ and $ZF=\gamma e^{\gamma u}$,
where $X=YZ^{-1}$.
Let us consider first the case $x_2\neq0$.
In this case, by Lemma \ref{lem:frac} a minimal fractional decomposition $X=YZ^{-1}$ for $X$ 
has $Z=\partial\circ\frac1{u''}\partial$ if $x_3\neq0$, and $Z=\partial$ if $x_3=0$.
In both cases the equation $ZF=\gamma e^{\gamma u}$
would imply $\gamma e^{\gamma u}\in\partial\mc V$, which is not the case by \eqref{20120907:eq1}.
In the case $x_2=0$,
a fractional decomposition for $X$ is
$X=YZ^{-1}$ given by \eqref{frac-liouv2}.
Hence, a solution $F\in\mc V$ to the equation $ZF=\gamma e^{\gamma u}$ 
is $F=e^{\gamma u}$,
and in this case we have 
$P=YF=\big(x_1\partial\circ\frac1{u'}\partial+x_3u'\big)e^{\gamma u}
=(x_1\gamma^2+x_3)e^{\gamma u}u'$.

For part (iii) it suffices to check, 
using the fractional decomposition \eqref{frac-liouv}, 
that $F=xu'-u\in\mc V$ is a solution of the equations
$$
\begin{array}{l}
\vphantom{\Big(}
ZF=\partial\frac{F'}{u''}=\frac{\delta}{\delta u}\tint u\,, \\
YF=\Big(x_1\partial^2\circ\frac1{u''}\partial+\frac{x_2+x_3(u')^2}{u''}\partial-x_3u'\Big)F=x_2x+x_3uu'\,.
\end{array}
$$

Similarly, for part (iv),
letting $F=-\sqrt{a_2+a_3(u')^2}\in\mc V$, we have
$$
\begin{array}{l}
\vphantom{\Big(}
ZF=\partial\frac{F'}{u''}=\frac{\delta}{\delta u}\tint \sqrt{a_2+a_3(u')^2}\,, \\
YF=-\Big(x_1\partial^2\circ\frac1{u''}\partial
+\frac{x_2+x_3(u')^2}{u''}\partial-x_3u'\Big)\sqrt{a_2+a_3(u')^2}\,.
\end{array}
$$

Next, let us prove part (v).
For $x_2\neq0$, 
consider the fractional decomposition $X=YZ^{-1}$ given by \eqref{frac-liouv}.
It is easy to check that, letting $F=\frac{u'}{x_2}$, we have
$$
\begin{array}{l}
ZF=\partial\circ\frac1{u''}\partial\frac{u'}{x_2}=0\,,\\
YF=\Big(x_1\partial^2\circ\frac1{u''}\partial
+\frac{x_2+x_3(u')^2}{u''}\partial-x_3u'\Big)\frac{u'}{x_2}=1\,.
\end{array}
$$
Hence, $\tint 0\ass{X}1$, as we wanted.
If $x_1\neq0,x_2=0,x_3\neq0$, a minimal fractional decomposition for $X$ is
$X=YZ^{-1}$ given by \eqref{frac-liouv2}.
Therefore the relation $\tint f\ass{X}1$ is equivalent
to the existence of $F\in\mc V$ such that
\begin{equation}\label{20120909:eq4}
ZF=\frac{F'}{u'}=\frac{\delta f}{\delta u}
\,\,,\,\,\,\,
YF=\Big(x_1\partial\circ\frac1{u'}\partial+x_3u'\Big)F=1\,.
\end{equation}
By Lemma \ref{20120909:lem1}(c),
the second equation in \eqref{20120909:eq4} implies that $F\in\mc V_0$ 
is such that $F'=\frac{\partial F}{\partial u}u'$.
In this case, the second equation in \eqref{20120909:eq4} reads
$$
x_1\partial\frac{\partial F}{\partial u}+x_3Fu'=1\,,
$$
which, by the commutation relation \eqref{eq:0.4},  is equivalent to 
$$
\Big(x_1\frac{\partial^2 F}{\partial u^2}+x_3F\Big)u'=1\,.
$$
But obviously, the above equation is never satisfied.
If $x_1=0,x_2=0,x_3\neq0$,
it is easy to check that $F=\frac1{x_3u'}$ solves
$$
YF=x_3u'\frac1{x_3u'}=1
\,\,,\,\,\,\,
ZF=\frac{1}{u'}\partial\frac1{x_3u'}=\frac{\delta}{\delta u}\tint\frac1{2x_3u'}\,,
$$
proving that $\tint\frac1{2x_3u'}\ass{X}1$.
Finally, $x_1\neq0,x_2=0,x_3=0$,
a minimal fractional decomposition $X=YZ^{-1}$ is given by $Y=x_1\partial$ and $Z=1$.
In this case, it is immediate to check that 
$F=\frac{x}{x_1}$ solves
$$
YF=x_1\partial\frac{x}{x_1}=1
\,\,,\,\,\,\,
ZF=\frac{x}{x_1}=\frac{\delta}{\delta u}\tint\frac{xu}{x_1}\,,
$$
proving that $\tint\frac{xu}{x_1}\ass{X}1$.

We are left to prove part (vi).
For $x_3\neq0$, 
consider the fractional decomposition $X=YZ^{-1}$ given by \eqref{frac-liouv}.
It is easy to check that, letting $F=\frac{-1}{x_3}$, we have
$$
\begin{array}{l}
ZF=\partial\circ\frac1{u''}\partial\frac{-1}{x_3}=0\,,\\
YF=\Big(x_1\partial^2\circ\frac1{u''}\partial
+\frac{x_2+x_3(u')^2}{u''}\partial-x_3u'\Big)\frac{-1}{x_3}=u'\,.
\end{array}
$$
Hence, $\tint 0\ass{X}u'$, as we wanted.
If $x_1\neq0,x_2\neq0,x_3=0$, the minimal fractional decomposition for $X$ is
$X=YZ^{-1}$ given by \eqref{frac-liouv1}.
Therefore the relation $\tint f\ass{X}u'$ is equivalent
to the existence of $F\in\mc V$ such that
\begin{equation}\label{20120909:eq5}
ZF=F'=\frac{\delta f}{\delta u}
\,\,,\,\,\,\,
YF=(x_1\partial^2+x_2)F=u'\,.
\end{equation}
If $F\in\mc V_n$, for $n\geq0$, we get, applying $\frac{\partial}{\partial u^{(n+2)}}$
to both sides of the second equation in \eqref{20120909:eq5}, 
that $\frac{\partial F}{\partial u^{(n)}}=0$.
Hence, it must be $F\in\mc F$. But in this case, the second equation in \eqref{20120909:eq5}
has clearly no solutions.
If $x_1=0,x_2\neq0,x_3=0$,
it is easy to check that $F=\frac{u'}{x_2}$ solves
$$
YF=x_2\frac{u'}{x_2}=u'
\,\,,\,\,\,\,
ZF=\partial\frac{u'}{x_2}=\frac{\delta}{\delta u}\tint\frac{-(u')^2}{2x_2}\,,
$$
proving that $\tint\frac{-(u')^2}{2x_2}\ass{X}u'$.
Finally, if $x_1\neq0,x_2=0,x_3=0$,
a minimal fractional decomposition $X=YZ^{-1}$ is given by $Y=x_1\partial$ and $Z=1$.
In this case, it is immediate to check that 
$F=\frac{u}{x_1}$ solves
$$
YF=x_1\partial\frac{u}{x_1}=u'
\,\,,\,\,\,\,
ZF=\frac{u}{x_1}=\frac{\delta}{\delta u}\tint\frac{u^2}{2x_1}\,,
$$
proving that $\tint\frac{u^2}{2x_1}\ass{X}u'$.
\end{proof}

In order the check orthogonality conditions \eqref{20120621:eq2}
for Liouville type integrable systems we will use the following results.
\begin{lemma}\label{20120909:lem2}
\begin{enumerate}[(a)]
\item
$(\mc C1)^\perp=\im(\partial)$;
\item
$(\mc Cu')^\perp=\im(\frac1{u'}\partial)$;
\item
$\big(\Span_{\mc C}\{1,u'\}\big)^\perp=\im(\partial\circ \frac1{u''}\partial)$;
\item
for $b_2,b_3\in\mc C\backslash\{0\}$, we have
$$
\big(\mc C\frac{\delta}{\delta u}\tint\sqrt{b_2+b_3(u')^2}\big)^\perp
=\im\Big(\frac{b_2+b_3(u')^2}{u''}\partial-b_3u'\Big)\,.
$$
\end{enumerate}
\end{lemma}
\begin{proof}
Parts (a) and (b) are immediate.
Let us prove part (c).
It is immediate to check, integrating by parts,
that $\tint (\alpha+\beta u')\partial\frac{f'}{u''}=0$ for every $\alpha,\beta\in\mc C$.
Hence, $\im(\partial\circ \frac1{u''}\partial)\subset\big(\Span_{\mc C}\{1,u'\}\big)^\perp$.
On the other hand, if
$f\in\big(\Span_{\mc C}\{1,u'\}\big)^\perp$,
it must be 
$$
f=\partial g=\frac1{u'}\partial h
\,\,\text{ for some }\, g,h\in\mc V\,.
$$
But then
$\partial h=u'\partial g=\partial(u'g)-u''g$,
which implies
$g=\frac1{u''}\partial(u'g-h)$.
Hence, 
$$
f=\partial \frac1{u''}\partial(u'g-h)\in\im(\partial\circ \frac1{u''}\partial)\,.
$$

We are left to prove part (d).
By definition of the variational derivative, we have
$$
-\frac1{b_3}\frac{\delta}{\delta u}\tint\sqrt{b_2+b_3(u')^2}
=\partial\frac{u'}{\sqrt{b_2+b_3(u')^2}}\,.
$$
The inclusion
$\im\Big(\frac{b_2+b_3(u')^2}{u''}\partial-b_3u'\Big)\subset
\Big(\mc C\partial\frac{u'}{\sqrt{b_2+b_3(u')^2}}\Big)^\perp$
follows by integration by parts, and the following straightforward identity
$$
\Big(\partial\circ\frac{b_2+b_3(u')^2}{u''}+b_3u'\Big)
\partial\frac{u'}{\sqrt{b_2+b_3(u')^2}}=0\,.
$$
We are left to prove the opposite inclusion.
If $f\in\Big(\mc C\partial\frac{u'}{\sqrt{b_2+b_3(u')^2}}\Big)^\perp$,
we have
$$
f=\frac{\partial g}{\partial\frac{u'}{\sqrt{b_2+b_3(u')^2}}}
=\frac{(b_2+b_3(u')^2)^{\frac32}}{b_2u''}\partial g
\,,
$$
for some $g\in\mc V$.
Letting $g=\frac{h}{\sqrt{b_2+b_3(u')^2}}$, we then get,
$$
\begin{array}{l}
\displaystyle{
f=\frac{(b_2+b_3(u')^2)^{\frac32}}{b_2u''}
\partial\frac{h}{\sqrt{b_2+b_3(u')^2}}
}\\
\displaystyle{
=\frac{(b_2+b_3(u')^2)^{\frac32}}{b_2u''}
\Big(
-\frac{b_3u'u''}{(b_2+b_3(u')^2)^{\frac32}}h
+\frac{1}{\sqrt{b_2+b_3(u')^2}}h'
\Big)
}\\
\displaystyle{
=
\Big(\frac{b_2+b_3(u')^2}{u''}\partial-b_3u'\Big)\frac{h}{b_2}\,.
}
\end{array}
$$
\end{proof}

\subsection{Integrability of the Lenard-Magri scheme in the case $b_1=0$}
\label{sec:3.3}

In this and the next two Sections we consider the case when $b_1=0$,
for which we get integrable Lenard-Magri schemes of S-type,
in the terminology introduced in Section \ref{sec:3.2}, in the case $a_1\neq0$ 
(described in Section \ref{sec:3.4}),
and of C-type in the case $a_1=0$ (described in Section \ref{sec:3.5}).
In Sections \ref{sec:3.6} we will consider the remaining case,
when $b_1\neq0$, for which we again get some integrable Lenard-Magri schemes of C-type.

According to Theorem \ref{th:lmscheme},
in order to apply successfully the Lenard-Magri scheme of integrability,
we need to find finite sequences $\{P_n\}_{n=0}^N$, $\{\tint h_n\}_{n=0}^N$,
satisfying all the relations \eqref{20120907:eq2}
and the orthogonality conditions \eqref{20120621:eq2}.
For $b_1=0$ we display below such sequences
separately in all possibilities for the coefficients $a_2,a_3,b_2,b_3$ being zero or non-zero,
and $a_1$ arbitrary.
(Note that, since we are assuming $b_1=0$, we don't need to consider the case $b_2=b_3=0$.)
\begin{enumerate}[(i)]
\item$b_2b_3\neq0,a_2a_3\neq0$:
$\tint0\ass{H}1\ass{K}\tint 0\ass{H}u'\ass{K}\tint\sqrt{b_2+b_3(u')^2}$.
\item$b_2b_3\neq0,a_2\neq0,a_3=0$:
$\tint0\ass{H}1\ass{K}\tint\sqrt{b_2+b_3(u')^2}$.
\item$b_2b_3\neq0,a_2=0,a_3\neq0$:
$\tint0\ass{H}u'\ass{K}\tint\sqrt{b_2+b_3(u')^2}$.
\item$b_2b_3\neq0,a_2=a_3=0$:
$\tint0\ass{H}0\ass{K}\tint\sqrt{b_2+b_3(u')^2}$.
\item$b_2\neq0,b_3=0,a_2a_3\neq0$:
$\tint0\ass{H}1\ass{K}\tint 0\ass{H}u'\ass{K}\tint \frac{-(u')^2}{2b_2}$.
\item$b_2\neq0,b_3=0,a_2\neq0,a_3=0$:
$\tint0\ass{H}1\ass{K}\tint 0$.
\item$b_2\neq0,b_3=0,a_2=0,a_3\neq0$:
$\tint0\ass{H}u'\ass{K}\tint \frac{-(u')^2}{2b_2}$.
\item$b_2\neq0,b_3=0,a_2=a_3=0$:
$\tint0\ass{H}0\ass{K}\tint 0$.
\item$b_2=0,b_3\neq0,a_2a_3\neq0$:
$\tint0\ass{H}u'\ass{K}\tint 0\ass{H}1\ass{K}\tint \frac1{2b_3u'}$.
\item$b_2=0,b_3\neq0,a_2\neq0,a_3=0$:
$\tint0\ass{H}1\ass{K}\tint \frac1{2b_3u'}$.
\item$b_2=0,b_3\neq0,a_2=0,a_3\neq0$:
$\tint0\ass{H}u'\ass{K}\tint 0$.
\item$b_2=0,b_3\neq0,a_2=a_3=0$:
$\tint0\ass{H}0\ass{K}\tint 0$.
\end{enumerate}
All the above $H$- and $K$-association relations hold due 
to Lemmas \ref{20120908:lem1} and \ref{20120908:lem2}.
Moreover, 
using Lemmas \ref{lem:frac} and \ref{20120909:lem2}
we check that both orthogonality conditions \eqref{20120621:eq2} hold.
Hence, by Theorem \ref{th:lmscheme}
all the above sequences can be continued indefinitely
to an infinite sequence
$$
\tint 0\ass{H}P_0\ass{K}\tint h_0\ass{H}P_1\ass{K}\tint h_1\ass{H}\dots\,.
$$
Note that, by Lemma \ref{20120907:lem1}, 
at each step the subsequent term is unique 
up to a linear combinations of the previous steps.

Next, we want to discuss integrability of the corresponding hierarchies of Hamiltonian equations
$\frac{du}{dt_n}=P_n,\,n\in\mb Z_+$.
Namely, we need to see when the vector spaces $\Span_{\mc C}\{\tint h_n\}\subset\mc V/\partial\mc V$
and $\Span_{\mc C}\{P_n\}\subset\mc V$ are infinite dimensional.

First, we consider the cases (vi), (viii), (xi) and (xii), where we show that
integrability does not occur (regardless of $a_1$ being zero or non-zero)
since the Lenard-Magri scheme repeats itself.
In case (vi), by Lemmas \ref{20120908:lem1} and \ref{20120908:lem2}
we have
$\mc H_0(H)=\mc C$ and, for every $\alpha\in\mc C$, 
$\{\tint f\in\mc F(K)\,|\,\tint f\ass{K}\alpha\}=\mc F_0(K)=\ker\big(\frac{\delta}{\delta u}\big)$.
Hence, any infinite sequence extending the given finite one will have
$\tint h_n\in\ker\big(\frac{\delta}{\delta u}\big)$ and $P_n\in\mc C$, for every $n\in\mb Z_+$.
Similarly, in case (xi) we have
$\mc H_0(H)=\mc Cu'$ and, for every $\alpha u'\in\mc Cu'$, 
$\{\tint f\in\mc F(K)\,|\,\tint f\ass{K}\alpha u'\}=\mc F_0(K)=\ker\big(\frac{\delta}{\delta u}\big)$.
Hence, any infinite sequence extending the given finite one will have
$\tint h_n\in\ker\big(\frac{\delta}{\delta u}\big)$ and $P_n\in\mc C u'$, for every $n\in\mb Z_+$.
In cases (viii) and (xii) we have 
$\mc H_0(H)=0$ and $\mc F_0(K)=\ker\big(\frac{\delta}{\delta u}\big)$.
Hence, 
$\tint h_n\in\ker\big(\frac{\delta}{\delta u}\big)$ and $P_n=0$, for every $n\in\mb Z_+$.
In conclusion, in all these cases $\Span_{\mc C}\{P_n\}$ is finite dimensional,
and integrability does not occur.

For the remaining 8 cases, 
we prove in Section \ref{sec:3.4} that 
when $a_1\neq0$ we get some integrable Lenard-Magri scheme of S-type,
and we prove in Section \ref{sec:3.5} that when $a_1=0$ 
we get some integrable Lenard-Magri scheme of C-type
(in the terminology of Section \ref{sec:3.2}).

\subsection{Integrable Lenard-Magri schemes of S-type: $b_1=0$ and $a_1\neq0$}
\label{sec:3.4}

\subsubsection*{Cases $(i), (ii), (iii), (iv)$}

In all the sequences (i)-(iv), after one or two steps, we arrive at (after shifting indices)
$\tint h_{-1}=\tint\sqrt{b_2+b_3(u')^2}$.
The next term in the sequence, which we denote $P_0$,
is obtained by solving the following equations for $F$ and $P_0$ in $\mc V$:
$$
\begin{array}{l}
\displaystyle{
BF=\partial\circ\frac1{u''}\partial F=\frac\delta{\delta u}\tint\tint\sqrt{b_2+b_3(u')^2}
\,,} \\
\displaystyle{
P_0=AF
=\Big(a_1\partial^2\circ\frac1{u''}\partial+\frac{a_2+a_3(u')^2}{u''}\partial-a_3u'\Big)F\,.
}
\end{array}
$$
It is easy to check that a solution is given by $F=-\sqrt{b_2+b_3(u')^2}$,
and
\begin{equation}\label{20120910:eq4}
\begin{array}{l}
\displaystyle{
P_0
=-\Big(\frac{a_1b_3u'}{\sqrt{b_2+b_3(u')^2}}\Big)^{\prime\prime}
+\frac{(a_3b_2-a_2b_3)}{\sqrt{b_2+b_3(u')^2}}
}\\
\displaystyle{
=
-\frac{a_1b_2b_3u'''}{(b_2+b_3(u')^2)^{\frac32}}
+3\frac{a_1b_2b_3^2u'(u'')^2}{(b_2+b_3(u')^2)^{\frac52}}
+\frac{(a_3b_2-a_2b_3)}{\sqrt{b_2+b_3(u')^2}}
\,.
}
\end{array}
\end{equation}

The above computation works regardless whether $a_1$ is zero or not.
But to prove that the Lenard-Magri scheme is integrable of S-type
we need to assume $a_1\neq0$,
in which case $\dord(P_0)=3$ is greater than
$\max\{\dord(A)-|H|+|K|,\,\dord(B)+|K|,\,\dord(C),\,\dord(D)+|K|\}$,
which is less than or equal to 2 (for all the cases (i)-(iv)).
By Proposition \ref{20120910:prop}
it follows that $\dord(P_n)=2n+3$ for every $n\geq0$.
In particular, the elements $P_n,\,n\in\mb Z_+$, are linearly independent.
Therefore,
by Theorem \ref{th:lmscheme} and Remark \ref{20120906:rem1},
each Hamiltonian PDE $\frac{du}{dt}=P_n,\,n\in\mb Z_+$,
is integrable, associated to an integrable Lenard-Magri scheme of S-type.
(Note that, since $\ker(B^*)\cap\ker(D^*)=\mc C\oplus\mc Cu'\neq0$,
we cannot conclude using Theorem \ref{th:lmscheme}
that $[P_m,P_n]=0$ for every $m,n\in\mb Z_+$,
and therefore that the sequence of equations $\frac{du}{dt_n}=P_n,\,n\in\mb Z_+$,
form a compatible hierarchy.
We conjecture, though, that this is the case.)

After rescaling $x$ and $t$ appropriately 
in the equation $\frac{du}{dt}=P_0$,
we conclude that the following bi-Hamiltonian equation
is integrable, associated to an integrable Lenard-Magri scheme of S-type:
\begin{equation}\label{beauty1}
\frac{du}{dt}
=
\frac{u'''}{(1+(u')^2)^{\frac32}}
-3\frac{u'(u'')^2}{(1+(u')^2)^{\frac52}}
+\frac{\alpha}{(1+(u')^2)^{\frac12}}
\,\,,\,\,\,\,\alpha\in\mc C\,.
\end{equation}
This is an equation of the form \cite[eq.(41.5)]{MSS90} with $a_3=(1+(u')^2)^{\frac12}$ 
(but in fact this particular $a_3$ does not appear in their list).

\subsubsection*{Cases $(v), (vii)$}

In the sequences (v) and (vii), after one or two steps, we arrive at (after shifting indices)
$\tint h_{-1}=\tint\frac{-(u')^2}{2b_2}$.
The next term in the sequence, which we denote $P_0$,
is obtained by solving the following equations for $F$ and $P_0$ in $\mc V$:
$$
\begin{array}{l}
\displaystyle{
BF=\partial\circ\frac1{u''}\partial F=\frac\delta{\delta u}\tint\tint\frac{-(u')^2}{2b_2}
\,,} \\
\displaystyle{
P_0=AF
=\Big(a_1\partial^2\circ\frac1{u''}\partial+\frac{a_2+a_3(u')^2}{u''}\partial-a_3u'\Big)F\,.
}
\end{array}
$$
It is easy to check that a solution is given by $F=\frac{(u')^2}{2b_2}$,
and
\begin{equation}\label{20120910:eq4b}
P_0=\frac{a_1}{b_2}u'''+\frac{a_2}{b_2}u'+\frac{a_3}{2b_2}(u')^3\,.
\end{equation}

As before, 
if $a_1\neq0$, $\dord(P_0)=3$ is greater than
$\max\{\dord(A)-|H|+|K|,\,\dord(B)+|K|,\,\dord(C),\,\dord(D)+|K|\}$,
which is at most 2.
By Proposition \ref{20120910:prop}
it follows that $\dord(P_n)=2n+3$ for every $n\geq0$.
In particular, the elements $P_n,\,n\in\mb Z_+$, are linearly independent.
Therefore,
by Theorem \ref{th:lmscheme} and Remark \ref{20120906:rem1},
each Hamiltonian PDE $\frac{du}{dt}=P_n,\,n\in\mb Z_+$,
is integrable, associated to an integrable Lenard-Magri scheme of S-type.

After rescaling $x$ and $t$ appropriately 
in the equation $\frac{du}{dt}=P_0$,
we conclude that the following bi-Hamiltonian equation
is integrable, associated to an integrable Lenard-Magri scheme of S-type:
\begin{equation}\label{beauty2}
\frac{du}{dt}
=
u'''+\epsilon u'+\alpha(u')^3\,,
\end{equation}
where $\epsilon$ is 1 (in case (v)) or 0 (in case (vii)) and $\alpha\in\mc C$.
By a Galilean transformation we can make $\epsilon=0$.
The resulting equation is called the potential modified KdV equation
(equation (4.11) in the list of \cite{MSS90}).

\subsubsection*{Cases $(ix), (x)$}

In all the sequences (ix) and (x), after one or two steps, we arrive at (after shifting indices)
$\tint h_{-1}=\tint\frac{1}{2b_3u'}$.
The next term in the sequence, which we denote $P_0$,
is obtained by solving the following equations for $F$ and $P_0$ in $\mc V$:
$$
\begin{array}{l}
\displaystyle{
BF=\partial\circ\frac1{u''}\partial F=\frac\delta{\delta u}\tint\tint\frac{1}{2b_3u'}
\,,} \\
\displaystyle{
P_0=AF
=\Big(a_1\partial^2\circ\frac1{u''}\partial+\frac{a_2+a_3(u')^2}{u''}\partial-a_3u'\Big)F\,.
}
\end{array}
$$
It is easy to check that a solution is given by $F=\frac{-1}{2b_3u'}$,
and
\begin{equation}\label{20120910:eq4c}
P_0=-\frac{a_1}{b_3}\frac{u'''}{(u')^3}+\frac{3a_1}{b_3}\frac{(u'')^2}{(u')^4}
+\frac{a_2}{2b_3}\frac1{(u')^2}+\frac{a_3}{b_3}
\,.
\end{equation}

If $a_1\neq0$, we have $\dord(P_0)=3$, which is greater than
$\max\{\dord(A)-|H|+|K|,\,\dord(B)+|K|,\,\dord(C),\,\dord(D)+|K|\}$,
which is at most 2.
By Proposition \ref{20120910:prop}
it follows that $\dord(P_n)=2n+3$ for every $n\geq0$.
In particular, the elements $P_n,\,n\in\mb Z_+$, are linearly independent.
Therefore,
by Theorem \ref{th:lmscheme} and Remark \ref{20120906:rem1},
each Hamiltonian PDE $\frac{du}{dt}=P_n,\,n\in\mb Z_+$,
is integrable, associated to an integrable Lenard-Magri scheme of S-type.

After rescaling $x$ and $t$ appropriately 
in the equation $\frac{du}{dt}=P_0$,
we conclude that the following bi-Hamiltonian equation 
is integrable, associated to an integrable Lenard-Magri scheme of S-type:
\begin{equation}\label{beauty3}
\frac{du}{dt}
=\frac{u'''}{(u')^3}-3\frac{(u'')^2}{(u')^4}
+\frac1{(u')^2}+\alpha
\,\,,\,\,\,\,
\alpha\in\mc C\,.
\end{equation}
As explained in \cite{MSS90}, by a point transformation one can reduce this equation
to an equation of the form (4.1.4) in their list.

\begin{remark}\label{20120911:rem1}
Note that equation $\frac{du}{dt}=P_0$ with $P_0$ given by \eqref{20120910:eq4c}
is transformed, by the hodograph transformation $x\mapsto u,\,u\mapsto-x$, 
to the equation with $P_0$ given by \eqref{20120910:eq4b}, after exchanging $a_2$ and $b_2$ 
with $a_3$ and $b_3$ respectively.
Equivalently, equation \eqref{beauty3} can be transformed to equation \eqref{beauty2}
up to rescaling of $x$ and $t$.
\end{remark}

\subsection{Integrable Lenard-Magri schemes of C-type with $a_1=b_1=0$}
\label{sec:3.5}

\subsubsection*{Cases $(i), (ii), (iii), (iv)$}

As pointed out above,
in all the sequences (i)-(iv) we arrive,
after one or two steps (and after shifting indices), at 
$\tint h_{-1}=\tint\sqrt{b_2+b_3(u')^2}$.
In the case $a_1=0$ we can actually find an explicit solution for the sequences
$\{\tint h_n\}_{n\in\mb Z_+}$ and $\{P_n\}_{n\in\mb Z_+}$
satisfying the recursive formulas 
\begin{equation}\label{20120912:eq2}
P_n\ass{K}\tint h_n\ass{H}P_{n+1}\,\,,\,\,\,\,\,n\in\mb Z_+\,.
\end{equation}
It is given by ($n\geq0$):
\begin{equation}\label{20120912:eq1}
\begin{array}{l}
\displaystyle{
P_n=
\sum_{k=0}^n\binom{n}{k}\frac{(2n-1-2k)!!}{(2n-2k)!!} \frac{\Delta^{n+1-k}a_3^k}{b_3^{n}}
\frac{-u'}{(b_2+b_3(u')^2)^{\frac12+n-k}}\,,
} \\
\displaystyle{
h_n=
\sum_{k=0}^n\binom{n}{k}\frac{(2n-1-2k)!!}{(2n-2k+2)!!} \frac{\Delta^{n+1-k}a_3^k}{b_3^{n+1}}
\frac{-u'}{(b_2+b_3(u')^2)^{\frac12+n-k}}\,,
}
\end{array}
\end{equation}
where $\Delta=a_2b_3-a_3b_2$, which is non-zero unless the operators $H$ and $K$ are proportional.
Here and further we let $(-1)!!=1$.

First, note that for $n=0$, the above expression for $P_0$ is the same 
as the one in \eqref{20120910:eq4} with $a_1=0$. 
Hence $\tint h_{-1}\ass{H}P_0$.
Next, we check that indeed the sequences $\{\tint h_n\}_{n\in\mb Z_+}$ and $\{P_n\}_{n\in\mb Z_+}$
solve the recursive relations \eqref{20120912:eq2}.
For this, we fix the fractional decompositions $H=AB^{-1}$ and $K=CD^{-1}$ given by 
$$
A=\frac{a_2+a_3(u')^2}{u''}\partial-a_3u'
\,\,,\,\,\,\,
C=\frac{b_2+b_3(u')^2}{u''}\partial-b_3u'
\,\,,
B=D=\partial\circ\frac1{u''}\partial\,.
$$
Since $B=D$, the relations \eqref{20120912:eq2} hold
if there exists an element $F_n$ such that
\begin{equation}\label{20120912:eq3}
CF_n=P_n
\,\,,\,\,\,\,
BF_n=\frac{\delta h_n}{\delta u}
\,\,,\,\,\,\,
P_{n+1}=AF_n\,.
\end{equation}
A solution $F_n$ to equations \eqref{20120912:eq3} is given by $F_n=-h_n$.
Since $h_n$ depends only on $u'$,
we have 
$$
B(-h_n)=-\partial\circ\frac1{u''}\partial h_n
=(-\partial)\frac{\partial h_n}{\partial u'}=\frac{\delta h_n}{\delta u}\,,
$$
hence $F_n=-h_n$ satisfies the second equation in \eqref{20120912:eq3}.
The first and third equations in \eqref{20120912:eq3} can be easily checked
using the following straightforward identities ($m\in\mb Z_+$):
$$
\begin{array}{l}
\displaystyle{
C\frac{1}{(b_2+b_3(u')^2)^{\frac m2}}=\frac{-b_3 m u'}{(b_2+b_3(u')^2)^{\frac m2}}
\,, } \\
\displaystyle{
A\frac{1}{(b_2+b_3(u')^2)^{\frac m2}}
=m\Delta\frac{-u'}{(b_2+b_3(u')^2)^{\frac m2 +1}}
+(m+1)a_3\frac{-u'}{(b_2+b_3(u')^2)^{\frac m2}}
\,. }
\end{array}
$$

If $\Delta\neq0$, all the elements $P_n$ are linearly independent,
therefore, by Theorem \ref{th:lmscheme} and Remark \ref{20120906:rem1},
each Hamiltonian PDE $\frac{du}{dt}=P_n,\,n\in\mb Z_+$,
is integrable, associated to an integrable Lenard-Magri scheme of C-type.

\subsubsection*{Cases $(v), (vii)$}

In the sequences (v) and (vii) we arrive,
after one or two steps (and after shifting indices), at 
$\tint h_{-1}=\tint\frac{-(u')^2}{2b_2}$.
In the case $a_1=0$ we can actually find an explicit solution for the sequences
$\{\tint h_n\}_{n\in\mb Z_+}$ and $\{P_n\}_{n\in\mb Z_+}$
satisfying the recursive formulas \eqref{20120912:eq2}.
It is given by ($n\in\mb Z_+$):
\begin{equation}\label{20120912:eq1b}
\begin{array}{l}
\displaystyle{
P_{n-1}=
\sum_{k=0}^{n}\binom{n}{k}\frac{(2k-1)!!}{(2k)!!} \frac{a_2^{n-k}a_3^k}{b_2^{n}}
(u')^{2k+1}\,,
} \\
\displaystyle{
h_{n-1}=
-\sum_{k=0}^{n}\binom{n}{k}\frac{(2k-1)!!}{(2k+2)!!} \frac{a_2^{n-k}a_3^k}{b_2^{n}}
(u')^{2k+2}\,.
}
\end{array}
\end{equation}

First, note that the above expression for $P_0$ is the same 
as the one in \eqref{20120910:eq4b} with $a_1=0$. 
Hence $\tint h_{-1}\ass{H}P_0$.
Next, we check that indeed the sequences $\{\tint h_n\}_{n\in\mb Z_+}$ and $\{P_n\}_{n\in\mb Z_+}$
solve the recursive relations \eqref{20120912:eq2}.
For this, we fix the fractional decompositions $H=AB^{-1}$ and $K=CD^{-1}$ given by 
$$
A=\frac{a_2+a_3(u')^2}{u''}\partial-a_3u'
\,\,,\,\,\,\,
B=\partial\circ\frac1{u''}\partial
\,\,,\,\,\,\,
C=b_2
\,\,,\,\,\,\,
D=\partial\,.
$$

The relations \eqref{20120912:eq2} hold
if there exist elements $F_n,G_n\in\mc V$ such that
\begin{equation}\label{20120912:eq3b}
CF_n=P_n
\,\,,\,\,\,\,
DF_n=\frac{\delta h_n}{\delta u}
\,\,,\,\,\,\,
BG_n=\frac{\delta h_n}{\delta u}
\,\,,\,\,\,\,
P_{n+1}=AG_n\,.
\end{equation}
Solutions $F_n,G_n$ to equations \eqref{20120912:eq3b} are given by 
$F_n=\frac1{b_2}P_n$ and $G_n=-h_n$.
The first and third equations in \eqref{20120912:eq3b} are immediate.
The third equation follows from the immediate identity 
$\frac{P_n}{b_2}=-\frac{\partial h_n}{\partial u'}$.
Finally, the fourth identity in \eqref{20120912:eq3b} is easily checked
using the Tartaglia-Pascal triangle.

Clearly, if $a_3\neq0$, all the elements $P_n$ are linearly independent,
therefore, by Theorem \ref{th:lmscheme} and Remark \ref{20120906:rem1},
each Hamiltonian PDE $\frac{du}{dt}=P_n,\,n\in\mb Z_+$,
is integrable, associated to an integrable Lenard-Magri scheme of C-type.

\subsubsection*{Cases $(ix), (x)$}

In the sequences (ix) and (x) we arrive,
after one or two steps (and shift of indices), at 
$\tint h_{-1}=\tint\frac{1}{2b_3u'}$.
If $a_1=0$ we can find an explicit solution for the sequences
$\{\tint h_n\}_{n\in\mb Z_+}$ and $\{P_n\}_{n\in\mb Z_+}$
satisfying the recursive formulas \eqref{20120912:eq2}.
It is given by ($n\in\mb Z_+$):
\begin{equation}\label{20120912:eq1c}
\begin{array}{l}
\displaystyle{
P_{n-1}=
\sum_{k=0}^{n}\binom{n}{k}\frac{(2k-1)!!}{(2k)!!} \frac{a_3^{n-k}a_2^k}{b_3^{n}}\frac1{(u')^{2k}}
\,, } \\
\displaystyle{
h_{n-1}=
\sum_{k=0}^{n}\binom{n}{k}\frac{(2k-1)!!}{(2k+2)!!} \frac{a_3^{n-k}a_2^k}{b_3^{n+1}}
\frac1{(u')^{2k+1}}
\,.}
\end{array}
\end{equation}

First, note that the above expression for $P_0$ is the same 
as the one in \eqref{20120910:eq4c} with $a_1=0$. 
Hence $\tint h_{-1}\ass{H}P_0$.
Next, we check that indeed the sequences $\{\tint h_n\}_{n\in\mb Z_+}$ and $\{P_n\}_{n\in\mb Z_+}$
solve the recursive relations \eqref{20120912:eq2}.
For this, we fix the fractional decompositions $H=AB^{-1}$ and $K=CD^{-1}$ given by 
$$
A=\frac{a_2+a_3(u')^2}{u''}\partial-a_3u'
\,\,,\,\,\,\,
B=\partial\circ\frac1{u''}\partial
\,\,,\,\,\,\,
C=b_3 u'
\,\,,\,\,\,\,
D=\frac1{u'}\partial\,.
$$
Since equations \eqref{20120912:eq3b} hold with
$F_n=\frac1{b_3u'}P_n$ and $G_n=-h_n$
(a fact that can be easily checked directly),
it follows that the recursive relations \eqref{20120912:eq2} hold.

Clearly, if $a_2\neq0$, all the elements $P_n$ are linearly independent,
therefore, by Theorem \ref{th:lmscheme} and Remark \ref{20120906:rem1},
each Hamiltonian PDE $\frac{du}{dt}=P_n,\,n\in\mb Z_+$,
is integrable, associated to an integrable Lenard-Magri scheme of C-type.

\subsection{Integrable Lenard-Magri scheme of C-type with $b_1\neq0$}
\label{sec:3.6}

As we did in the previous sections, we study here the integrability of the Lenard-Magri scheme
when $b_1\neq0$.
We will consider separately the various cases, depending on the parameters $b_2,b_3,a_2,a_3$
being zero or non-zero.

\subsubsection*{Case 1: $b_2b_3\neq0$}

Let us consider first the case when $b_2$ and $b_3$ are both non-zero.
If $\tint h_0\in\mc V/\partial\mc V$ and $P_0\in\mc V$ satisfy
the relations $\tint 0\ass{H}P_0\ass{K}\tint h_0$,
then, by Lemmas \ref{20120908:lem1} and \ref{20120908:lem2},
we necessarily have $P_0\in\mc C\oplus\mc Cu'$
and $\frac{\delta h_0}{\delta u}=0$.
Hence, any infinite sequence extending the given finite one will have
$\tint h_n\in\ker\big(\frac{\delta}{\delta u}\big)$ and $P_n\in\mc C\oplus\mc Cu'$, 
for every $n\in\mb Z_+$.
In other words, 
the Lenard-Magri scheme repeats itself
and integrability does not occur.

\subsubsection*{Case 2: $b_2\neq0,b_3=0,a_3=0$}

In the case when $b_1b_2\neq0$, $b_3=0$ and $a_3=0$,
we can find explicitly all possible solutions for the sequences
$\{\tint h_n\}_{n\in\mb Z_+}$ and $\{P_n\}_{n\in\mb Z_+}$
satisfying the Lenard-Magri recursive relations \eqref{20120907:eq3}.

In order to describe such solutions, we need to introduce some polynomials.
We let $p_n(x;A,\epsilon),\,q_n(x;A,\epsilon)$, $n\in\mb Z_+$,
be the sequences of polynomials, depending on the 
$2\times2$ matrix 
$A=\Big(\begin{array}{ll} a_1 & a_2 \\ b_1 & b_2 \end{array}\Big)$,
and on the sequence of constant parameters $\epsilon=(\epsilon_0,\epsilon_1,\dots)$,
defined by the following recursive relations:
$p_0(x;A,\epsilon)=0$, and
\begin{equation}\label{20121003:eq1}
\begin{array}{l}
\displaystyle{
p_{n+1}(x;A,\epsilon)
=\frac{a_1}{b_1}p_n(x;A,\epsilon)
+\frac{a_2b_1-a_1b_2}{b_1^2}q_n(x;A,\epsilon)
} \\
\displaystyle{
\Big(\frac{d^2}{dx^2}+2b_{12}\frac{d}{dx}\Big)q_n(x;A,\epsilon)=p_n(x;A,\epsilon)
\,.}
\end{array}
\end{equation}
Here and further, as before, we use the notation \eqref{notation} with $x_i$ and $x_j$ 
replaced by $b_i$ and $b_j$.
It is easy to see that, if $p(x)$ is a polynomial of degree $n$,
then a solution $q(x)$ of the differential equation
$q''(x)+2b_{12}q'(x)=p(x)$
is a polynomial of degree $n+1$,
defined uniquely up to an additive constant $\epsilon_0$.
Hence, at each step in the recursion \eqref{20121003:eq1},
the resulting polynomial $p_{n+1}(x)$ depends on the previous step $p_n(x)$
and on the choice of a constant parameter $\epsilon_{n+1}$.

With the above notation, all sequences 
$\{\tint h_n\}_{n\in\mb Z_+}$, $\{P_n\}_{n\in\mb Z_+}$,
satisfying the Lenard-Magri recursive relations \eqref{20120907:eq3},
are as follows:
\begin{equation}\label{20121003:eq2}
\begin{array}{l}
\displaystyle{
\vphantom{\Bigg(}
P_n=
p_n(x;A,\epsilon^+)e^{b_{12}x}
+p_n(-x;A,\epsilon^-)e^{-b_{12}x}+a_2\delta_n\,,
} \\
\displaystyle{
\vphantom{\Bigg(}
h_n=
\frac1{b_1} 
\Big(q_n^\prime(x;A,\epsilon^+)+b_{12}q_n(x;A,\epsilon^+)\Big)
e^{b_{12}x} u
} \\
\displaystyle{
\vphantom{\Bigg(}
\,\,\,\,\,\,\,\,\,
-\frac1{b_1} 
\Big(q_n^\prime(-x;A,\epsilon^-)+b_{12}q_n(-x;A,\epsilon^-)\Big)
e^{-b_{12}x} u
\,,
}
\end{array}
\end{equation}
where $\epsilon^\pm=(\epsilon^\pm_0,\epsilon^\pm_1,\dots)$
and $\delta=(\delta_0,\delta_1,\dots)$
are arbitrary sequences of constant parameters.

It is not hard to check that 
the sequences $\{\tint h_n\}_{n\in\mb Z_+}$ and $\{P_n\}_{n\in\mb Z_+}$
indeed solve the recursive relations \eqref{20120907:eq3},
and any solution of the recursive relations \eqref{20120907:eq3} is obtained in this way.
To conclude, we observe that,
since $\Delta=a_2b_1-a_1b_2$ is non-zero
(unless the operators $H$ and $K$ are proportional),
all the elements $P_n$ are linearly independent,
therefore, by Theorem \ref{th:lmscheme} and Remark \ref{20120906:rem1},
each Hamiltonian PDE $\frac{du}{dt}=P_n,\,n\in\mb Z_+$,
is integrable, associated to an integrable Lenard-Magri scheme of C-type.

\subsubsection*{Case 3: $b_2\neq0,b_3=0,a_3\neq0$}

In the case when $b_1b_2\neq0$, $b_3=0$ and $a_3\neq0$,
we have by Lemma \ref{20120908:lem1}(b) that 
$\mc H_0(H)=a_2\mc C\oplus\mc Cu'$.
On the other hand, by Lemma \ref{20120908:lem2}(vi) there is no
element $\tint f\in\mc F(K)$ such that $\tint f\ass{K}u'$.
Similarly, 
by Lemma \ref{20120908:lem1}(a) we have 
$\mc F_0(K)=\mc C\tint e^{b_{12}x}u
+\mc C\tint e^{-b_{12}x}u
+\ker\big(\frac\delta{\delta u}\big)$.
On the other hand, by Lemma \ref{20120908:lem2}(i) there is no
element $P\in\mc F(H)$ such that $\tint e^{\pm b_{12}x}u\ass{H}P$.

In conclusion,
the Lenard-Magri recursion scheme, in this case,
cannot be applied,
since the following finite sequences cannot be extended
to infinite sequences satisfying the conditions \eqref{20120907:eq3}:
$$
\tint 0\ass{H} u'\ass{K}\not\exists\tint f
\,\,,\,\,\,\,
\tint 0\ass{H} 0\ass{K}\tint e^{\pm b_{12}x}u
\ass{H}\not\exists P
\,.
$$

\subsubsection*{Case 4: $b_2=0,b_3\neq0,a_2=0$}

In the case when $b_1b_3\neq0$, $b_2=0$ and $a_2=0$,
we can find explicitly all possible solutions for the sequences
$\{\tint h_n\}_{n\in\mb Z_+}$ and $\{P_n\}_{n\in\mb Z_+}$
satisfying the Lenard-Magri recursive relations \eqref{20120907:eq3}.
They are as follows:
\begin{equation}\label{20121003:eq3}
\begin{array}{l}
\displaystyle{
\vphantom{\Bigg(}
P_n=
p_n(u;A,\epsilon^+)e^{b_{13}u}
+p_n(-u;A,\epsilon^-)e^{-b_{13}u}+a_3\delta_nu'\,,
} \\
\displaystyle{
\vphantom{\Bigg(}
\frac{\delta h_n}{\delta u}=
\frac1{b_1} 
\Big(q_n^\prime(u;A,\epsilon^+)+b_{13}q_n(u;A,\epsilon^+)\Big)
e^{b_{13}u}
} \\
\displaystyle{
\vphantom{\Bigg(}
\,\,\,\,\,\,\,\,\,
-\frac1{b_1} 
\Big(q_n^\prime(-u;A,\epsilon^-)+b_{13}q_n(-u;A,\epsilon^-)\Big)
e^{-b_{12}u}
\,,
}
\end{array}
\end{equation}
where 
$p_n(u;A,\epsilon^+)$ and $q_n(u;A,\epsilon^+)$
are the polynomials defined in \eqref{20121003:eq1},
depending on the 
matrix $A=\Big(\begin{array}{ll} a_1 & a_3 \\ b_1 & b_3 \end{array}\Big)$,
and on the sequences of constant parameters
$\epsilon^\pm=(\epsilon^\pm_0,\epsilon^\pm_1,\dots)$
and $\delta=(\delta_0,\delta_1,\dots)$.

It is not hard to check, as in case 1, that 
the sequences $\{\tint h_n\}_{n\in\mb Z_+}$ and $\{P_n\}_{n\in\mb Z_+}$
solve the recursive relations \eqref{20120907:eq3},
and any solution of the recursive relations \eqref{20120907:eq3} is obtained in this way.
We also observe that,
since $\Delta=a_3b_1-a_1b_3$ is non-zero
(unless the operators $H$ and $K$ are proportional),
all the elements $P_n$ are linearly independent,
therefore, by Theorem \ref{th:lmscheme} and Remark \ref{20120906:rem1},
each Hamiltonian PDE $\frac{du}{dt}=P_n,\,n\in\mb Z_+$,
is integrable, associated to an integrable Lenard-Magri scheme of C-type.

\subsubsection*{Case 5: $b_2=0,b_3\neq0,a_2\neq0$}

In the case when $b_1b_3\neq0$, $b_2=0$ and $a_3=0$,
we have by Lemma \ref{20120908:lem1}(b) that 
$\mc H_0(H)=\mc C\oplus a_3\mc Cu'$.
On the other hand, by Lemma \ref{20120908:lem2}(v) there is no
element $\tint f\in\mc F(K)$ such that $\tint f\ass{K}1$.
Similarly, 
by Lemma \ref{20120908:lem1}(a) we have 
$\mc F_0(K)=\mc C\tint e^{b_{13}u}
+\mc C\tint e^{-b_{13}u}
+\ker\big(\frac\delta{\delta u}\big)$.
On the other hand, by Lemma \ref{20120908:lem2}(ii) there is no
element $P\in\mc F(H)$ such that $\tint e^{\pm b_{13}x}u\ass{H}P$.

In conclusion,
the Lenard-Magri recursion scheme, in this case,
cannot be applied,
since the following finite sequences cannot be extended
to infinite sequences satisfying the conditions \eqref{20120907:eq3}:
$$
\tint 0\ass{H} 1\ass{K}\not\exists\tint f
\,\,,\,\,\,\,
\tint 0\ass{H} 0\ass{K}\tint e^{\pm b_{13}u}
\ass{H}\not\exists P
\,.
$$

\subsubsection*{Case 6: $b_2=b_3=0$}

In the case when $b_1\neq0$, which we set equal to $1$, and $b_2=b_3=0$,
we have different possibilities according to the constants $a_2$ and $a_3$ being zero or not.

If $a_2a_3\neq0$,
the Lenard-Magri recursion scheme cannot be applied.
Indeed, by Lemma \ref{20120908:lem1} we have 
$\mc H_0(H)=\mc C\oplus\mc Cu'$
and
$\mc F_0(K)=\mc C\tint u\oplus\ker\big(\frac{\delta}{\delta u}\big)$,
and, whichever way we start the finite sequences
$\{\tint h_n\}_{n=0}^N$, $\{P_n\}_{n=0}^N$ as in \eqref{20120907:eq2},
there is no way to extend them to non-trivial infinite sequences:
$$
\begin{array}{l}
\displaystyle{
\vphantom{\bigg(}
\tint 0\ass{H} 1\ass{K}\tint xu\ass{H}\not\exists P_1
\,,} \\
\displaystyle{
\vphantom{\bigg(}
\tint 0\ass{H} u'\ass{K}\tint\frac12 u^2\ass{H}\not\exists P_1
\,,} \\
\displaystyle{
\vphantom{\bigg(}
\tint 0\ass{H} 0\ass{K}\tint u\ass{H} a_2x+a_3uu' \ass{K} \not\exists \tint h_1
\,.} 
\end{array}
$$

Next, we consider the cases when exactly one element between $a_2$ and $a_3=0$ is zero.
Recall the sequence of polynomials $p_n(x;A,\epsilon)$ defined by the
recursive equations \eqref{20121003:eq1}.
In the case $b_1=1,b_2=0$, such equations reduce to
$p_0(x;a_1,a_2,\epsilon)=0$ and
\begin{equation}\label{20121003:eq1b}
p_{n+1}(x;a_1,a_2,\epsilon)
=\Big(\frac{a_1}{b_1}+\frac{a_2}{b_1}\Big(\frac{d}{dx}\Big)^{-2}\Big)p_n(x;a_1,a_2,\epsilon)
\,.
\end{equation}
Here $\Big(\frac{d}{dx}\Big)^{-2}\Big)$ means integrating twice with respect to $x$,
which is defined uniquely up to adding a linear term $\epsilon_{2n}+\epsilon_{2n+1}x$.
In particular, at each step the degree increases by two.

In the case $a_2\neq0,a_3=0$, 
it is not hard to prove that all the sequences 
$\{\tint h_n\}_{n\in\mb Z_+}$, $\{P_n\}_{n\in\mb Z_+}$,
satisfying the Lenard-Magri recursive relations \eqref{20120907:eq3},
are as follows:
\begin{equation}\label{20121003:eq2b}
P_n=p_n(x;a_1,a_2,\epsilon)+\delta_n
\,\,,\,\,\,\,
\tint h_n=\tint \Big(\frac{d}{dx}\Big)^{-1}p_n(x;a_1,a_2,\epsilon)u\,,
\end{equation}
where $\epsilon=(\epsilon_0,\epsilon_1,\dots)$ and $(\delta_0,\delta_1,\dots)$
are arbitrary sequence of constant parameters.
Since, obviously,
all the elements $P_n$ are linearly independent,
we conclude that each Hamiltonian PDE $\frac{du}{dt}=P_n,\,n\in\mb Z_+$,
is integrable, associated to an integrable Lenard-Magri scheme of C-type.

Similarly, in the case $a_2=0,a_3\neq0$, all the sequences 
$\{\tint h_n\}_{n\in\mb Z_+}$, $\{P_n\}_{n\in\mb Z_+}$,
satisfying the Lenard-Magri recursive relations \eqref{20120907:eq3},
are as follows:
\begin{equation}\label{20121003:eq2c}
P_n=p_n(u;a_1,a_2,\epsilon)u'+\delta_nu'
\,\,,\,\,\,\,
\tint h_n=\tint \Big(\frac{d}{du}\Big)^{-2}p_n(u;a_1,a_2,\epsilon)\,,
\end{equation}
where $\epsilon=(\epsilon_0,\epsilon_1,\dots)$ and $(\delta_0,\delta_1,\dots)$
are arbitrary sequences of constant parameters.
Again, we conclude that each Hamiltonian PDE $\frac{du}{dt}=P_n,\,n\in\mb Z_+$,
is integrable, associated to an integrable Lenard-Magri scheme of C-type.

\subsection{Summary}
\label{sec:3.7}

Let us summarize the results from the previous sections 
by listing all the possibilities for the pairs $H$ and $K$ as in \eqref{20121006:eq2},
and specifying,
using the terminology of Section \ref{sec:3.2},
whether the corresponding Lenard-Magri sequence
\begin{equation}\label{20121006:eq4}
\tint0\ass{H}P_0\ass{K}\tint h_0\ass{H} P_1\ass{K}\tint h_1\ass{H}
P_2\ass{K}\dots\,,
\end{equation}
is \emph{integrable of S-type},
i.e. the $\mc C$-span of the elements $P_n$'s and $\tint h_n$'s is infinite dimensional
and $H$ has order strictly greater than $K$ ($b_1=0,a_1\neq0$),
whether it is \emph{integrable of C$_1$-type},
i.e. the $\mc C$-span of the elements $P_n$'s and $\tint h_n$'s is infinite dimensional
and the orders of $H$ and $K$ are both equal to $-1$ ($b_1=a_1=0$),
whether it is \emph{integrable of C$_2$-type},
i.e. the $\mc C$-span of the elements $P_n$'s and $\tint h_n$'s is infinite dimensional
and $H$ has order less than or equal to $K$ and $K$ has order 1 ($b_1\neq0$),
whether it is of \emph{finite type},
i.e. the $\mc C$-span of the elements $P_n$'s or $\tint h_n$'s is necessarily finite dimensional,
or whether it is \emph{blocked},
i.e. there are choices of $\tint h_n$ or $P_n$ for which the scheme
cannot be continued.
This is the list of all possibilities:
\begin{itemize}
\item
\emph{integrable of S-type}:
\begin{enumerate}[(a)]
\item
$b_1=0$, $(b_2,b_3)\neq(0,0)$, $a_1a_2a_3\neq0$;
\item
$b_1=0$, $a_1\neq0$, 
and either $b_2\neq0$, $a_2=0$ and $(b_3,a_3)\neq(0,0)$,
or $b_3\neq0$, $a_3=0$ and $(b_2,a_2)\neq(0,0)$.
\end{enumerate}
\item
\emph{integrable of C$_1$-type}:
\begin{enumerate}[]
\item
$b_1=a_1=0$, 
and either $b_2a_3\neq0$ and $(b_3,a_2)$ arbitrary,
or $b_3a_2\neq0$ and $(b_2,a_3)$ arbitrary.
\end{enumerate}
\item
\emph{integrable of C$_2$-type}:
\begin{enumerate}[(a)]
\item
$b_1a_1\neq0$, 
and either $b_2=a_2=0$ and $(b_3,a_3)\neq(0,0)$,
or $b_3=a_3=0$ and $(b_2,a_2)\neq(0,0)$;
\item
$b_1\neq0$, $a_1=0$,
and either $b_2=a_2=0$ and $a_3\neq0$ (with $b_3$ arbitrary),
or $b_3=a_3=0$ and $a_2\neq0$ (with $b_2$ arbitrary).
\end{enumerate}
\item
\emph{finite type}:
\begin{enumerate}[(a)]
\item
$b_1b_2b_3\neq0$, $a_1=0$, $(a_2,a_3)\neq(0,0)$;
\item
$b_1=0$, $a_1\neq0$,
and either $b_2=a_2=0$ and $b_3\neq0$ (with $a_3$ arbitrary),
or $b_3=a_3=0$ and $b_2\neq0$ (with $a_2$ arbitrary).
\item[(c1)]
$b_1=a_1=0$,
and either $b_2=a_2=0$ and $b_3a_3\neq0$,
or $b_2a_2\neq0$ and $b_3=a_3=0$;
\item[(c2)]
$b_1b_2b_3\neq0, a_1a_2a_3\neq0$;
\item[(d)]
$b_1b_2b_3\neq0, a_1\neq0, a_2a_3=0$.
\end{enumerate}
\item
\emph{blocked}:
\begin{enumerate}[(a)]
\item
$b_1\neq0$, $a_1=0$, 
and either $b_2=0$, $a_2\neq0$ and $(b_3,a_3)\neq(0,0)$,
or $b_3=0$, $a_3\neq0$ and $(b_2,a_2)\neq(0,0)$;
\item
$b_1\neq0, b_2b_3=0, a_1a_2a_3\neq0$;
\item
$b_1a_1\neq0$, 
and either $b_2a_3\neq0$ and $b_3=a_2=0$,
or $b_2=a_3=0$ and $b_3a_2\neq0$.
\end{enumerate}
\end{itemize}

\subsection{Going to the left}
\label{sec:3.8}

Suppose we have an integrable Lenard-Magri sequence \eqref{20121006:eq4}
(of S, C$_1$ or $C_2$-type).
A natural question is whether this sequence can be continued to the left:
\begin{equation}\label{20121006:eq1}
\dots\ass{H}P_{-1}\ass{K}\tint0\ass{H}P_0\ass{K}\tint h_0\ass{H} P_1\ass{K}\tint h_1\ass{H}
P_2\ass{K}\dots\,.
\end{equation}
In this way we get some additional equations
compatible with the given hierarchy $\frac{du}{dt_n}=P_n,\,n\in\mb Z_+$,
and additional integrals of motion $\tint h_n,\,n=-1,-2,\dots$,
in involution with the given $\tint h_n$'s, with $n\geq0$.

Clearly, trying to extend the Lenard-Magri scheme \eqref{20121006:eq1} to the left
amounts to switching the roles of the non-local Poisson structures $H$ and $K$,
and to constructing the ``dual'' Lenard-Magri sequence
\begin{equation}\label{20121006:eq3}
\tint0\ass{K}P_{-1}\ass{H}\tint h_{-1}\ass{K} P_{-2}\ass{H}\tint h_{-2}\ass{K}
P_{-3}\ass{H}\dots\,.
\end{equation}
So, we need to study, for each possible choice of the parameters $a_i,b_i,i=1,2,3$,
what type of Lenard-Magri scheme we get when we switch all the 
coefficients $a_i$'s with the $b_i$'s.

By looking at the list of all possibilities in the previous section,
after switching the roles of $H$ and $K$ we have the following:
\begin{itemize}
\item
integrable of S-type (a) $\stackrel{H\leftrightarrow K}{\longleftrightarrow}$ finite type (a);
\item
integrable of S-type (b) $\stackrel{H\leftrightarrow K}{\longleftrightarrow}$ blocked (a);
\item
integrable of C$_1$-type $\stackrel{H\leftrightarrow K}{\longleftrightarrow}$ integrable of C$_1$-type;
\item
integrable of C$_2$-type (a) $\stackrel{H\leftrightarrow K}{\longleftrightarrow}$ 
integrable of C$_2$-type (a);
\item
integrable of C$_2$-type (b) $\stackrel{H\leftrightarrow K}{\longleftrightarrow}$ 
finite-type (b);
\item
finite-type (c) $\stackrel{H\leftrightarrow K}{\longleftrightarrow}$ finite-type (c);
\item
finite-type (d) $\stackrel{H\leftrightarrow K}{\longleftrightarrow}$ blocked (b);
\item
blocked (c) $\stackrel{H\leftrightarrow K}{\longleftrightarrow}$ blocked (c).
\end{itemize}


We are only interested in the integrable (S or C-type) Lenard-Magri schemes.
We see from the above list that, after exchanging the roles of $H$ and $K$,
three things can happen.
The ``dual'' Lenard-Magri scheme \eqref{20121006:eq3} can be of \emph{finite}-type
(this happens in the cases S(a) and C$_2$(b)).
In this situation continuing the Lenard-Magri scheme to the left 
we never get any new interesting integrals of motion or equations.


The second possibility is that the ``dual'' Lenard-Magri scheme \eqref{20121006:eq3} 
is of \emph{integrable}-type
(this happens in the cases C$_1$ and C$_2$(a)).
In this situation we can continue the Lenard-Magri scheme to the left indefinitely.
In other words, in each of these cases we can merge two integrable systems,
``dual'' to each other, to get one integrable system with twice as many integrals of motion
and equations.


The most interesting situation is when the ``dual'' Lenard-Magri scheme \eqref{20121006:eq3} 
is \emph{blocked} (which happens in the case S(b)).
In this case,
if the sequence \eqref{20121006:eq3} is blocked at $P_{-k},\,k\geq0$,
we obtain an integrable PDE which is not of evolutionary type.

\subsection{Non-evolutionary integrable equations}
\label{sec:3.9}

According to the previous discussion, 
we need to consider the case of integrable Lenard-Magri scheme
of S-type (b), which means the following 5 cases:
\begin{enumerate}
\item
$b_1=0, b_2\neq0, b_3\neq0, a_1\neq0, a_2\neq0, a_3=0$; 
\item
$b_1=0, b_2\neq0, b_3\neq0, a_1\neq0, a_2=0, a_3\neq0$; 
\item
$b_1=0, b_2\neq0, b_3\neq0, a_1\neq0, a_2=0, a_3=0$; 
\item
$b_1=0, b_2\neq0, b_3=0, a_1\neq0, a_2=0, a_3\neq0$; 
\item
$b_1=0, b_2=0, b_3\neq0, a_1\neq0, a_2\neq0, a_3=0$; 
\end{enumerate}

\subsubsection*{Case 1: $b_1=0, b_2\neq0, b_3\neq0, a_1\neq0, a_2\neq0, a_3=0$} 

This case gives, to the right, the Lenard-Magri scheme listed as case (ii) in Section \ref{sec:3.4},
while, after exchanging the roles of $H$ and $K$ we get, to the left,
the ``blocked'' Lenard-Magri scheme listed as case 3 in Section \ref{sec:3.6}.
Hence, overall, we get the following scheme:
$$
\begin{array}{l}
\vphantom{\Bigg(}
\displaystyle{
\not\exists P\ass{K}
\tint e^{\pm a_{12}x}u
\ass{H}0 \ass{K}\tint 0
\ass{H}1\ass{K}\tint\sqrt{b_2+b_3(u')^2}\ass{H}
} \\
\displaystyle{
\ass{H}
-\frac{a_1b_2b_3u'''}{(b_2+b_3(u')^2)^{\frac32}}
+3\frac{a_1b_2b_3^2u'(u'')^2}{(b_2+b_3(u')^2)^{\frac52}}
-\frac{a_2b_3}{\sqrt{b_2+b_3(u')^2}}
\ass{K}\dots
\,.
}
\end{array}
$$
Here and further $\pm$ means that we take arbitrary linear combination of the above expressions
with $+$ and with $-$.
Trying to solve naively for $P$ in the above scheme,
we get the following expression
$$
P=
\pm \frac{b_2}{a_{12}} e^{\pm a_{12}x}
+b_3u' \partial^{-1} \big(e^{\pm a_{12}x}u'\big)
\,.
$$
The meaning of the above expression for $P$
is that the following partial differential equation
is a member of the integrable hierarchy associated to the Lenard-Magri scheme of S-type (ii):
\begin{equation}\label{20121020:eq3}
\Big(\frac{u_t}{u_x}\Big)_x
=
\pm\frac{b_2}{a_{12}}
\Big(\frac1{u_x}e^{\pm a_{12}x}\Big)_x
+b_3 e^{\pm a_{12}x} u_x
\,.
\end{equation}

\subsubsection*{Case 2: $b_1=0, b_2\neq0, b_3\neq0, a_1\neq0, a_2=0, a_3\neq0$} 

This case gives, to the right, the Lenard-Magri scheme listed as case (iii) in Section \ref{sec:3.4},
while, after exchanging the roles of $H$ and $K$ we get, to the left,
the ``blocked'' Lenard-Magri scheme listed as case 5 in Section \ref{sec:3.6}.
Hence, overall, we get the following scheme:
$$
\begin{array}{l}
\vphantom{\Bigg(}
\displaystyle{
\not\exists P\ass{K}
\tint e^{\pm a_{13}u}
\ass{H}0 \ass{K}\tint 0
\ass{H}u'\ass{K}\tint\sqrt{b_2+b_3(u')^2}\ass{H}
} \\
\displaystyle{
\ass{H}
-\frac{a_1b_2b_3u'''}{(b_2+b_3(u')^2)^{\frac32}}
+3\frac{a_1b_2b_3^2u'(u'')^2}{(b_2+b_3(u')^2)^{\frac52}}
+\frac{a_3b_2}{\sqrt{b_2+b_3(u')^2}}
\ass{K}\dots
\,.
}
\end{array}
$$
Trying to solve naively for $P$ we get
$$
P=
\pm a_{13}b_2\partial^{-1}e^{\pm a_{13}u}
+b_3u'e^{\pm a_{13}u}
\,.
$$
This means that the following hyperbolic partial differential equation
is a member of the integrable hierarchy associated to the Lenard-Magri scheme of S-type (iii):
\begin{equation}\label{20121020:eq4}
u_{tx}=
\pm a_{13}b_2 e^{\pm a_{13}u}\pm\frac{b_3}{a_{13}} \big(e^{\pm a_{13}u}\big)_{xx}
\,.
\end{equation}

\subsubsection*{Case 3: $b_1=0, b_2\neq0, b_3\neq0, a_1\neq0, a_2=0, a_3=0$} 

This case gives, to the right, the Lenard-Magri scheme listed as case (iv) in Section \ref{sec:3.4},
while, after exchanging the roles of $H$ and $K$ we get, to the left,
the ``blocked'' Lenard-Magri scheme listed as case 6 in Section \ref{sec:3.6}.
Hence, overall, we get, depending on how we choose to continue the scheme to the left,
the following two possibilities (or any their linear combination):
$$
\begin{array}{l}
\vphantom{\Bigg(}
\displaystyle{
\not\exists P
\ass{K}
\frac{1}{a_1}\tint xu
\ass{H}
1
\ass{K}\tint 0\ass{H}0
\ass{K}\tint\sqrt{b_2+b_3(u')^2}\ass{H}
} \\
\displaystyle{
\ass{H}
-\frac{a_1b_2b_3u'''}{(b_2+b_3(u')^2)^{\frac32}}
+3\frac{a_1b_2b_3^2u'(u'')^2}{(b_2+b_3(u')^2)^{\frac52}}
\ass{K}\dots
\,,
}
\end{array}
$$
or
$$
\begin{array}{l}
\vphantom{\Bigg(}
\displaystyle{
\not\exists P
\ass{K}
\frac{1}{a_1}\tint \frac12 u^2
\ass{H}
u'
\ass{K}\tint 0\ass{H}0
\ass{K}\tint\sqrt{b_2+b_3(u')^2}\ass{H}
} \\
\displaystyle{
\ass{H}
-\frac{a_1b_2b_3u'''}{(b_2+b_3(u')^2)^{\frac32}}
+3\frac{a_1b_2b_3^2u'(u'')^2}{(b_2+b_3(u')^2)^{\frac52}}
\ass{K}\dots
\,.
}
\end{array}
$$
Trying to solve naively for $P$ we get, in the first case
$$
P=
\frac{b_2}{2a_1}x^2+\frac{b_3}{a1}xuu'-\frac{b_3}{a_1}u'\partial^{-1}u
\,,
$$
which corresponds to the following integrable non-evolutionary partial differential equation:
\begin{equation}\label{20121020:eq1}
\Big(\frac{u_t}{u_x}\Big)_x
=
\frac{b_2}{2a_1}\Big(\frac{x^2}{u_x}\Big)_x+\frac{b_3}{a_1}xu_x
\,.
\end{equation}
In the second case we get
$$
P=
\frac{b_2}{a_1}\partial^{-1}u+\frac{b_3}{2a_1}u^2u'
\,,
$$
which corresponds to the following integrable hyperbolic partial differential equation:
\begin{equation}\label{20121020:eq2}
u_{tx}
=
\frac{b_2}{a_1}u+\frac{b_3}{6a_1}(u^3)_{xx}
\,.
\end{equation}
In conclusion, both equations \eqref{20121020:eq1} and \eqref{20121020:eq2}
are members of the integrable hierarchy associated to the Lenard-Magri scheme of S-type (iv).

\subsubsection*{Case 4: $b_1=0, b_2\neq0, b_3=0, a_1\neq0, a_2=0, a_3\neq0$} 

This case gives, to the right, the Lenard-Magri scheme listed as case (vii) in Section \ref{sec:3.4},
while, after exchanging the roles of $H$ and $K$ we get, to the left,
the ``blocked'' Lenard-Magri scheme listed as case 5 in Section \ref{sec:3.6}.
Hence, overall, we get the following scheme:
$$
\begin{array}{l}
\vphantom{\Bigg(}
\displaystyle{
\not\exists P\ass{K}
\tint e^{\pm a_{13}u}
\ass{H}0 \ass{K}\tint 0
\ass{H}u'\ass{K}\tint\frac{-(u')^2}{2b_2}\ass{H}
} \\
\displaystyle{
\ass{H}
\frac{a_1}{b_2}u'''+\frac{a_3}{2b_2}(u')^3
\ass{K}\dots
\,.
}
\end{array}
$$
Trying to solve naively for $P$ we get
$$
P=
\pm a_{13}b_2\partial^{-1}e^{\pm a_{13}u}
\,.
$$
This means that the following hyperbolic partial differential equation
is a member of the integrable hierarchy associated to the Lenard-magri scheme of S-type (vii):
\begin{equation}\label{20121020:eq6}
u_{tx}=
\pm a_{13}b_2 e^{\pm a_{13}u}
\,.
\end{equation}
As expected, equation \eqref{20121020:eq6} is obtained by \eqref{20121020:eq4}
letting $b_3=0$.

\subsubsection*{Case 5: $b_1=0, b_2=0, b_3\neq0, a_1\neq0, a_2\neq0, a_3=0$} 

This case gives, to the right, the Lenard-Magri scheme listed as case (x) in Section \ref{sec:3.4},
while, after exchanging the roles of $H$ and $K$ we get, to the left,
the ``blocked'' Lenard-Magri scheme listed as case 3 in Section \ref{sec:3.6}.
Hence, overall, we get the following scheme:
$$
\begin{array}{l}
\vphantom{\Bigg(}
\displaystyle{
\not\exists P\ass{K}
\tint e^{\pm a_{12}x}u
\ass{H}0 \ass{K}\tint 0
\ass{H}1\ass{K}\tint\frac1{2b_3u'}\ass{H}
} \\
\displaystyle{
\ass{H}
-\frac{a_1}{b_3}\frac{u'''}{(u')^3}+\frac{3a_1}{b_3}\frac{(u'')^2}{(u')^4}
+\frac{a_2}{2b_3}\frac1{(u')^2}
\ass{K}\dots
\,.
}
\end{array}
$$
Trying to solve naively for $P$ in the above scheme,
we get the following expression
$$
P=
b_3u' \partial^{-1} \big(e^{\pm a_{12}x}u'\big)
\,,
$$
and the associated non-evolutionary partial differential equation is
\begin{equation}\label{20121020:eq5}
\Big(\frac{u_t}{u_x}\Big)_x
=
b_3 e^{\pm a_{12}x} u_x
\,.
\end{equation}
In conclusion, equation \eqref{20121020:eq5}
is a member of the integrable hierarchy associated to the Lenard-Magri scheme of S-type (x).
Note that this equation is obtained letting $b_2=0$ in equation \eqref{20121020:eq3}.

\subsubsection*{Conclusion} 

After rescaling the variables $u$, $x$ and $t$, or replacing $x$ by $x+$ const., or $u$ by $u+$ const.,
in equations \eqref{20121020:eq3}-\eqref{20121020:eq5},
we conclude that the following are all the integrable non-evolutionary partial differential equations
which are members of some integrable hierarchy of bi-Hamiltonian equations, 
with $H$ and $K$ as in \eqref{20121006:eq2}:
\begin{eqnarray}
&& u_{tx}
=
e^{u}-\alpha e^{-u}
+\epsilon(e^{u}-\alpha e^{-u})_{xx}
\,,\label{20121020:eq8}\\
&& \Big(\frac{u_t}{u_x}\Big)_x
=
(e^{x}-\alpha e^{-x}) u_x
+\epsilon \Big(\frac{e^{x}-\alpha e^{-x}}{u_x}\Big)_x
\,, \label{20121020:eq7}\\
&& u_{tx}
=
u+(u^3)_{xx}
\,, \label{20121020:eq10}\\
&& \Big(\frac{u_t}{u_x}\Big)_x
=
\Big(\frac{x^2}{u_x}\Big)_x+xu_x
\,,\label{20121020:eq9}
\end{eqnarray}
where $\alpha$ and $\epsilon$ are $0$ or $1$.

Recall that the case when $\epsilon=0$ equation \eqref{20121020:eq8}
is the Liouville equation when $\alpha=0$,
and the sinh-Gordon equation when $\alpha=1$, cf. \cite{Dor93}.
Equation \eqref{20121020:eq7} (respectively \eqref{20121020:eq9}) can be obtained from equation \eqref{20121020:eq8} (resp. \eqref{20121020:eq10})
by the hodograph transformation $u\mapsto x$, $x\mapsto -u$.
Equation \eqref{20121020:eq10} is called the ``short pulse equation'', \cite{SW02},
and its integrability was proved in \cite{SS04}.
Equations \eqref{20121020:eq8} with $\epsilon=1$
seems to be new.

\section{KN type integrable system}
\label{sec:4}

In this section $\mc V$ is a field of differential functions in $u$,
and, as usual, we assume that $\mc V$ contains all the functions that we encounter
in our computations.

Recall from \cite[Example 4.6]{DSK12} that the following
is a pair of compatible non-local Poisson structures:
$$
L_1=
u'\partial^{-1}\circ u' 
\,\,\text{ (Sokolov) }
\,\,,\,\,\,\,
L_2=
\partial^{-1}\circ u'\partial^{-1}\circ u'\partial^{-1}
\,\,\text{ (Dorfman) }
\,.
$$
We consider two non-local Poisson structures $H$ and $K$
which are linear combinations of $L_1$ and $L_2$:
$H=a_1L_1+a_2L_2$ and $K=b_1L_1+b_2L_2$.
As we have seen in the example of Liouville type integrable systems, 
discussed in Section \ref{sec:3}, 
integrable hierarchies associated to Lenard-Magri schemes of C type are usually 
not very interesting (cf. Sections \ref{sec:3.5} and \ref{sec:3.6}).
Hence, in this section, we will only consider integrable Lenard-Magri schemes of S-type
(in the terminology of Section \ref{sec:3.2}),
which is possible only when the order of the pseudodifferential operator $H$
is greater that the order of $K$,
namely when $a_1\neq0$ and $b_1=0$.
Therefore, we consider the following compatible pair of non-local structures:
\begin{equation}\label{20121006:eq2kn}
H=u'\partial^{-1}\circ u' + a\partial^{-1}\circ u'\partial^{-1}\circ u'\partial^{-1}
\,\,,\,\,\,\,
K=\partial^{-1}\circ u'\partial^{-1}\circ u'\partial^{-1}\,,
\end{equation}
with $a\in\mc C$.
We want to discuss the integrability of the corresponding Lenard-Magri scheme.

\subsection{Preliminary computations}
\label{sec:4.1}

Note that $K$ is the inverse of a differential operator, hence its minimal fractional decomposition is
$K=1D^{-1}$, where
\begin{equation}\label{frac-D}
D=\partial\circ\frac1{u'}\partial\circ\frac1{u'}\partial\,.
\end{equation}
We next find a minimal fractional decomposition for $H$.
It is given by the following
\begin{lemma}\label{lem:frac-kn}
For every $a\in\mc C$, we have $H=AB^{-1}$, where
\begin{equation}\label{frac-kn}
\begin{array}{l}
\displaystyle{
A=
\bigg(
\partial^2-2\frac{u''}{u'}\partial+\Big(\frac{u''}{u'}\Big)^\prime+a
\bigg)\circ
\frac{1}{D(u')}
\partial-u'
\,,}\\
\displaystyle{
B=
\partial\circ\frac1{u'}\partial\circ\frac1{u'}\partial\circ
\frac{1}{D(u')}
\partial
\,.} 
\end{array}
\end{equation}
Here and further, we have, recalling \eqref{frac-D},
\begin{equation}\label{20121015:eq3}
D(u')=\bigg(\frac1{u'}\Big(\frac{u''}{u'}\Big)^\prime\bigg)^\prime\,.
\end{equation}
The above fractional decomposition is minimal only for $a\neq0$.
For $a=0$, the minimal fractional decomposition for $H$ is 
$H=1S^{-1}$, where
\begin{equation}\label{frac-kn1}
S=\frac1{u'}\partial\circ\frac1{u'}\,.
\end{equation}
\end{lemma}
\begin{proof}
We need to prove that $AB^{-1}=S^{-1}+aD^{-1}$.
By looking at the coefficient of $a$ in $AB^{-1}$, we get
$$
\frac{1}{D(u')}
\partial
\Bigg(
\partial\circ\frac1{u'}\partial\circ\frac1{u'}\partial\circ
\frac{1}{D(u')}
\partial
\Bigg)^{-1}
=
\partial^{-1}
\circ u'
\partial^{-1}
\circ u'
\partial^{-1}
=D^{-1}\,.
$$
Letting $a=0$ in $AB^{-1}$, we have
$$
\begin{array}{l}
\displaystyle{
\Bigg(
\bigg(
\partial^2-2\frac{u''}{u'}\partial+\Big(\frac{u''}{u'}\Big)^\prime
\bigg)\circ
\frac{1}{D(u')}
\partial-u'
\Bigg)
\Bigg(
\partial\circ\frac1{u'}\partial\circ\frac1{u'}\partial\circ
\frac{1}{D(u')}
\partial
\Bigg)^{-1}
} \\
\displaystyle{
=\bigg(
\partial^2-2\frac{u''}{u'}\partial+\Big(\frac{u''}{u'}\Big)^\prime
-u'\partial^{-1}\circ D(u')
\bigg)\circ
\partial^{-1}u'\partial^{-1}\circ u'\partial^{-1}
} \\
\displaystyle{
=
\partial\circ u'\partial^{-1}\circ u'\partial^{-1}
-2u''\partial^{-1}\circ u'\partial^{-1}
+\Big(\frac{u''}{u'}\Big)^\prime\partial^{-1}u'\partial^{-1}\circ u'\partial^{-1}
} \\
\displaystyle{
-u'\partial^{-1}\circ D(u')\partial^{-1}u'\partial^{-1}\circ u'\partial^{-1}
=
u'\partial^{-1}\circ u'
+\bigg(
u'\partial^{-1}\circ \frac{u''}{u'}
} \\
\displaystyle{
-u''\partial^{-1}
+\Big(\frac{u''}{u'}\Big)^\prime\partial^{-1}u'\partial^{-1}
-u'\partial^{-1}\circ D(u')\partial^{-1}u'\partial^{-1}
\bigg)\circ u'\partial^{-1}
\,.}
\end{array}
$$
In the last identity we used the Leibniz rule for $\partial$:
$\partial\circ f=f\partial+f'$.
To conclude the proof, we need to check that the expression in parenthesis
in the RHS is zero:
\begin{equation}\label{20121025:eq1}
u'\partial^{-1}\circ \frac{u''}{u'}
-u''\partial^{-1}
+\Big(\frac{u''}{u'}\Big)^\prime\partial^{-1}u'\partial^{-1}
-u'\partial^{-1}\circ D(u')\partial^{-1}u'\partial^{-1}
=0\,.
\end{equation}
This identity is obtained applying repeatedly the commutation relation ($f\in\mc V$),
\begin{equation}\label{20121025:eq2}
\partial^{-1}\circ f=f\partial^{-1}-\partial^{-1}\circ f'\partial^{-1}\,,
\end{equation}
which is a consequence of the Leibniz rule for $\partial$,
and using the expression \eqref{20121015:eq3} for $D(u')$.
\end{proof}

In order to apply successfully the Lenard-Magri scheme of integrability we need to compute
the kernel of the operator $B$.
\begin{lemma}\label{20121025:lem1}
The kernel of the operator $B$ in \eqref{frac-kn}
is a 4-dimensional vector space over $\mc C$,
spanned by
$$
\begin{array}{l}
\displaystyle{
f_1=1
\,\,,\,\,\,\,
f_2=\frac1{u'}\Big(\frac{u''}{u'}\Big)^\prime
\,\,,\,\,\,\,
f_3=\frac{u}{u'}\Big(\frac{u''}{u'}\Big)^\prime-\frac{u''}{u'}
\,,} \\
\displaystyle{
f_4=
\frac{u^2}{u'}\Big(\frac{u''}{u'}\Big)^\prime
-2u\frac{u''}{u'}+2u'
\,.}
\end{array}
$$
\end{lemma}
\begin{proof}
It is immediate to check that all the elements $f_i$ are indeed in the kernel of $B$.
On the other hand, since $B$ has order 4, its kernel has dimension at most 4.
\end{proof}

\subsection{Applying the Lemard-Magri scheme for $a\neq0$}
\label{sec:4.2}

According to the Lenard-Magri scheme of integrability,
starting with $\tint h_{-1}=\tint0$,
we need to find sequences $\{\tint h_n\}_{n=0}^N$
and $\{P_n\}_{n=0}^N$
solving the recursion conditions \eqref{20120907:eq2}.


Since $C=1$, in order to find solutions of the scheme \eqref{20120907:eq2} for $N=3$,
we need to find elements $F_n,h_n,P_n\in\mc V,\,n=0,\dots,3$,
such that
$$
BF_{n}=\frac{\delta h_{n-1}}{\delta u}
\,\,,\,\,\,\,
P_n=AF_{n}
\,\,,\,\,\,\,
\frac{\delta h_n}{\delta u}=DP_n
\,,
$$
for all $n=0,1,2,3$ (we let, as usual, $\tint h_{-1}=\tint 0$).
Recalling the expressions \eqref{frac-D} and \eqref{frac-kn} of $A,B,D$,
and using Lemma \ref{20121025:lem1},
it is straightforward but lengthy calculation to find solutions:
$$
\begin{array}{lll}
\displaystyle{
F_0=\frac{f_2}a
=\frac1{au'}\Big(\frac{u''}{u'}\Big)^\prime
\,\,,}&
P_0=1
\,\,,&
\tint h_0=\tint0
\,, \\
\displaystyle{
F_1=\frac{f_3}a 
= \frac{u}{au'}\Big(\frac{u''}{u'}\Big)^\prime-\frac{u''}{u'}
\,\,,}&
P_1=u
\,\,,&
\tint h_1=\tint0
\,, \\
\displaystyle{
F_2=\frac{f_4}a 
=\frac{u^2}{au'}\Big(\frac{u''}{u'}\Big)^\prime-2\frac{uu''}{au'}+\frac{2}{a}u'
\,\,,}&
P_2=u^2
\,\,,&
\tint h_2=\tint0
\,, \\
\displaystyle{
F_3=-f_1=-1
\,\,,}&
P_3=u'
\,\,,&
\displaystyle{
\tint h_3=\frac12\int\Big(\frac{u''}{u'}\Big)^2
\,.} 
\end{array}
$$
Hence, we get the following Lenard-Magri scheme
\begin{equation}\label{20121102:eq1}
\begin{array}{l}
\displaystyle{
\tint 0\ass{H}1\ass{K}\tint 0
\ass{H} u\ass{K}\tint 0
\ass{H}u^2\ass{K}\tint 0\ass{H}
} \\
\displaystyle{
\ass{H}u'\ass{K}\frac12\tint\Big(\frac{u''}{u'}\Big)^2
\ass{K}\dots\,.
}
\end{array}
\end{equation}


We next prove that the scheme \eqref{20121102:eq1} can be extended indefinitely.
According to Theorem \ref{th:lmscheme},
this is the case, provided that the orthogonality conditions \eqref{20120621:eq2} hold.
Since $C=1$, the first condition in \eqref{20120621:eq2} is trivial.
As for the second orthogonality condition,
let $\varphi\in\big(\Span_{\mc C}\{P_0,P_1,P_2,P_3\}\big)^\perp$.
Since $\varphi\perp P_0$, we have that $\varphi=\varphi_1^\prime$, for some $\varphi_1\in\mc V$.
Since $\varphi\perp P_1$, we have that $\varphi_1=\frac{\varphi_2^\prime}{u'}$, 
for some $\varphi_2\in\mc V$.
Since $\varphi\perp P_2$, we have that $\varphi_2=\frac{\varphi_3^\prime}{u'}$, 
for some $\varphi_3\in\mc V$.
And, finally,
since $\varphi\perp P_3$, we have that $\varphi_3=\frac{\varphi_4^\prime}{D(u')}$, 
for some $\varphi_4\in\mc V$.
In conclusion, $\varphi=B\varphi_4$, proving the second orthogonality condition \eqref{20120621:eq2}.


We compute explicitly the next element $P_4$ in the Lenard scheme,
which gives the first non-trivial equation of the corresponding bi-Hamiltonian hierarchy.
For this, we need to solve, for $F_4,P_4\in\mc V$, the following equations
$$
BF_{4}=\frac{\delta h_{3}}{\delta u}
=D(u')
\,\,,\,\,\,\,
P_4=AF_{4}
\,.
$$
The general solution is:
$$
F_{4}=
\Big(\frac{u''}{u'}\Big)^\prime-\frac12\Big(\frac{u''}{u'}\Big)^2
+(a-\alpha_1)f_1+\frac{\alpha_2}{a}f_2
+\frac{\alpha_3}{a} f_3
+\frac{\alpha_4}{a} f_4
\,,
$$
where $\alpha_i,\,i=1,\dots,4$, are arbitrary constants.
Hence, the first non-trivial integrable equation in the hierarchy has the form:
\begin{equation}\label{20121102:eq2}
\frac{du}{dt}=P_4=
u'''-\frac32\frac{(u'')^2}{u'}+\alpha_1 u'
+\alpha_2+\alpha_3u+\alpha_4u^2\,.
\end{equation}


In order to prove that equation \eqref{20121102:eq2} is indeed integrable,
we are left to prove that the sequences $\{\tint h_n\}_{n\in\mb Z_+}$
and $\{P_n\}_{n\in\mb Z_+}$ are linearly independent.
For this, we use Proposition \ref{20120910:prop}.

Since $\dord(D(u'))=4$,
we have $\dord(A)=6$, $\dord(B)=7$, $\dord(C)=-\infty$ and $\dord(D)=3$.
Moreover, we clearly have $|H|=-1$ and $|K|=-3$.
Hence, the RHS of inequality \eqref{20120911:eq1} is $4$.

Next, we compute the differential order of the next element $P_5$ in the Lenard-Magri scheme.
It is obtained by solving, for $\xi_4=\frac{\delta h_4}{\delta u}, F_5, P_5\in\mc V$,
the following equations:
\begin{equation}\label{20121102:eq3}
\xi_4=DP_4
\,\,,\,\,\,\,
BF_5=\xi_4
\,\,,\,\,\,\,
P_5=AF_5\,.
\end{equation}
From the first equation in \eqref{20121102:eq3} we get
$$
\xi_4=\partial\frac1{u'}\partial\frac1{u'}\partial(u'''+\rho)
\,,
$$
where $\rho\in\mc V$ has $\dord(\rho)=2$.
Hence, the second equation in \eqref{20121102:eq3} gives
$$
\frac1{D(u')}\partial F_5=u'''+\rho_1\,,
$$
where $\dord(\rho_1-\rho)=0$.
In particular, $F_5$ has differential order less than or equal to $3$.
It follows by the third equation in \eqref{20121102:eq3} that
$$
P_5=
\bigg(
\partial^2-2\frac{u''}{u'}\partial+\Big(\frac{u''}{u'}\Big)^\prime+a
\bigg)(u'''+\rho_1)
-u'F_5\,.
$$
Hence, $\frac{\partial P_5}{\partial u^{(5)}}=1$, and $\frac{\partial P_5}{\partial u^{(n)}}=0$ 
for every $n>5$.
In particular, $\dord(P_5)=5$.

According to Proposition \ref{20120910:prop},
since $\dord(P_5)=5>4$,
we have that
$\dord\big(\frac{\delta h_n}{\delta u}\big)=2n-2$,
and $\dord(P_n)=2n-5$,
for every $n\geq3$.
In particular,
all the elements $\{\tint h_n\}_{n\in\mb Z_+}$
and $\{P_n\}_{n\in\mb Z_+}$ are linearly independent.
As a consequence, every equation of the hierarchy $\frac{du}{dt_n}=P_n,\,n\in\mb Z_+$,
including equation \eqref{20121102:eq2}, is integrable.

Note that, since the kernels of $B^*$ and $D^*$ have non-zero intersections,
we cannot conclude that $[P_m,P_n]$ is zero for every $m,n\in\mb Z_+$
(cf. Theorem \ref{th:lmscheme}),
but we believe this is true.

When all constants $\alpha_i$ are equal to zero,
equation \eqref{20121102:eq2} is called in \cite{Dor93}
the Krichever-Novikov (KN) equation.
As explained in \cite{MS12}, equation \eqref{20121102:eq2} can be reduced to
the KN equation by some point transformation.

\subsection{The case $a=0$}
\label{sec:4.3}

In the case when $a=0$ all the computations are much easier.
Since $A=C=1$, 
the recursive conditions $\tint h_{n-1}\ass{H}P_n\ass{K}\tint h_n,\,n\in\mb Z_+$,
are equivalent to the equations
$$
BP_n=\frac{\delta h_{n-1}}{\delta u}
\,\,,\,\,\,\,
\frac{\delta h_n}{\delta u}=DP_n
\,.
$$
It is easy to find the first few steps of the Lenard-Magri scheme:
\begin{equation}\label{20121102:eq1b}
\tint 0\ass{H}P_0=u'\ass{K}\tint h_0=\frac12\tint\Big(\frac{u''}{u'}\Big)^2
\ass{H}P_1=u'''-\frac{3}{2}\frac{(u'')^2}{u'}+\alpha_1 u'
\ass{K}\dots
\end{equation}
for arbitrary $\alpha_1\in\mc C$.


As before, the scheme \eqref{20121102:eq1b} can be extended indefinitely.
Indeed, since $C=1$, the first orthogonality condition in \eqref{20120621:eq2} is trivial,
while the second one holds since $P_0^\perp=\im B$.


Moreover, in this case $\dord(A)=\dord(C)=-\infty$, $\dord(B)=2$, and $\dord(D)=3$,
so the RHS of inequality \eqref{20120911:eq1} is $0$.
Since $\dord(P_0)=1>0$,
we can apply Proposition \ref{20120910:prop}
to deduce that all the elements $\{\tint h_n\}_{n\in\mb Z_+}$
and $\{P_n\}_{n\in\mb Z_+}$ are linearly independent.

In conclusion, 
every equation of the hierarchy $\frac{du}{dt_n}=P_n,\,n\in\mb Z_+$,
is integrable.
Note that the first non-trivial equation is $\frac{du}{dt}=P_1$,
which is the same as equation \eqref{20121102:eq2} with $\alpha_2=\alpha_3=\alpha_4=0$.

\subsection{One step back}
\label{sec:4.4}

As we did in the example of Liouville type integrable systems,
we can ask whether the Lenard-Magri schemes \eqref{20121102:eq1} and \eqref{20121102:eq1b}
can be continued to the left.

This amounts to finding $P_{n}\in\mc V$ and $\tint h_{n}\in\mc V/\partial\mc V$, with $n\leq-1$,
such that
\begin{equation}\label{20121103:eq1}
\dots\ass{K}\tint h_{-2}\ass{H}P_{-1}\ass{K}\tint0
\end{equation}


We consider separately the cases $a\neq0$ and $a=0$.
When $a\neq0$, the conditions \eqref{20121103:eq1}
give the following equations
for $P_{-1}$, $F$ and $\xi_{-2}=\frac{\delta h_{-2}}{\delta u}$:
\begin{equation}\label{20121103:eq2}
DP_{-1}=0
\,\,,\,\,\,\,
AF=P_{-1}
\,\,,\,\,\,\,
\xi_{-2}=BF\,,
\end{equation}
where $A$, $B$, and $D$ are as in \eqref{frac-D} and \eqref{frac-kn}.
All solutions $P_{-1}$ of the first equation in \eqref{20121103:eq2} are
$$
P_{-1}=c_0+c_1u+c_2u^2\,,
$$
with $c_0,c_1,c_2\in\mc C$.
Next, we want to find all solutions $F$ of the second equation in \eqref{20121103:eq2}.
Applying $\frac{\partial}{\partial u^{(n)}}$, with $n\geq4$, to both sides of the equation $AF=P_{-1}$
we immediately get that $\dord(F)\leq3$ and $\partial F=fD(u')$, with $\dord(f)\leq1$.
Hence, the second equation in \eqref{20121103:eq2} can be rewritten as
the following system of equations,
\begin{equation}\label{20121103:eq3}
\begin{array}{l}
\displaystyle{
\partial^2f-2\frac{u''}{u'}\partial f+\Big(\frac{u''}{u'}\Big)^\prime f+af-u'F
=c_0+c_1u+c_2u^2\,,
} \\
\displaystyle{
\partial F=fD(u')
\,,}
\end{array}
\end{equation}
for $F,f\in\mc V$ with $\dord(F)\leq3$ and $\dord(f)\leq1$.
Applying $\frac{\partial}{\partial u^{(3)}}$ to both sides of the first equation in \eqref{20121103:eq3}
and $\frac{\partial}{\partial u^{(4)}}$ to both sides of the second equation in \eqref{20121103:eq3},
we get
$$
\frac{\partial F}{\partial u'''}=\frac{f}{(u')^2}
\,\,,\,\,\,\,
\frac{\partial f}{\partial u'}=0\,.
$$
Hence, $\dord(f)\leq0$.
Next, 
applying $\frac{\partial}{\partial u^{(2)}}$ to the first equation in \eqref{20121103:eq3}
and $\frac{\partial}{\partial u^{(3)}}$ to the second equation in \eqref{20121103:eq3},
we get
$$
\frac{\partial F}{\partial u''}=-2\frac{u''}{(u')^3}f-\frac{1}{(u')^2}\partial f
\,\,,\,\,\,\,
\partial f=\frac{\partial f}{\partial u}u'\,.
$$
Hence, $f$ is a function of $u$ only.
Using the above result, 
we can rewrite the second equation in \eqref{20121103:eq3},
after integrating by parts twice, as
$$
\partial F=
\partial\Big(
\frac{f}{u'}\Big(\frac{u''}{u'}\Big)^\prime
-\frac{\partial f}{\partial u}\frac{u''}{u'}
+\frac{\partial^2 f}{\partial u^2} u'
\Big)
-\frac{\partial^3 f}{\partial u^3}(u')^2\,.
$$
In particular, it must be $\frac{\partial^3 f}{\partial u^3}(u')^2\in\partial\mc V$,
which is possible only if $\frac{\partial^3 f}{\partial u^3}(u')^2=0$,
see \cite{BDSK09}.
In conclusion, $f$ must be a quadratic polynomial in $u$ with constant coefficients,
and 
$F=\frac{f}{u'}\Big(\frac{u''}{u'}\Big)^\prime
-\frac{\partial f}{\partial u}\frac{u''}{u'}
+\frac{\partial^2 f}{\partial u^2} u'+$ const.
Plugging these results back into equation \eqref{20121103:eq3}
we finally get that
$$
\frac{\partial F}{D(u')}=f=\frac{c_0}a+\frac{c_1}au+\frac{c_2}au^2\,.
$$
Hence, the third equation in \eqref{20121103:eq2} gives $\xi_{-2}=0$.
In conclusion, in this case
the ``dual'' Lenard-Magri sequence, obtained by exchanging the roles of $H$ and $K$,
is of \emph{finite} type,
namely it repeats itself with $\tint h_n\in\ker\big(\frac\delta{\delta u}\big)$
and $P_n\in\ker(D)$ for every $n\leq-1$,
and we don't get any new interesting integrals of motion or equations.


Next, we consider the case $a=0$.
In this case, the conditions \eqref{20121103:eq1}
give the following equations
for $P_{-1}$, $P_{-2}$, and $\tint h_{-2}$:
\begin{equation}\label{20121103:eq4}
DP_{-1}=0
\,\,,\,\,\,\,
DP_{-2}=\frac{\delta h_{-2}}{\delta u}=SP_{-1}\,,
\end{equation}
where $S$, and $D$ are as in \eqref{frac-D} and \eqref{frac-kn1}.
As before, 
$P_{-1}=c_0+c_1u+c_2u^2$, with $c_0,c_1,c_2\in\mc C$.
Hence, the second equation in \eqref{20121103:eq4} reads
\begin{equation}\label{20121103:eq5}
\bigg(\frac1{u'}
\Big(\frac{\partial P_{-2}}{u'}\Big)^\prime\bigg)^\prime
=
\frac1{u'}\Big(
\frac{c_0+c_1u+c_2u^2}{u'}
\Big)^\prime
\,.
\end{equation}
For every $n\in\mb Z_+$, we have the identity
$$
\frac1{u'}\Big(
\frac{u^n}{u'}
\Big)^\prime
=\frac12\partial\frac{u^n}{(u')^2}
+\frac n2\frac{u^{n-1}}{u'}\,.
$$
It follows that the RHS of \eqref{20121103:eq5} cannot be a total derivative
unless $c_1=c_2=0$ (cf. \cite{BDSK09}).
Moreover, if $c_1=c_2=0$ equation \eqref{20121103:eq5} reduces to
$$
\Big(\frac{\partial P_{-2}}{u'}\Big)^\prime
=
\frac{c_0}{2u'}+\text{ const.}u'
\,,
$$
which, for the same reason as before, has no solutions unless $c_0=0$.
In conclusion,
for every non-zero $P_{-1}\in\ker D$,
the ``dual'' Lenard-Magri scheme
is \emph{blocked} at $P_{-2}$.
In this case, as we saw in Section \ref{sec:3.9},
we obtain integrable PDE's which are not of evolutionary type.

In particular, for $(c_2,c_3)\neq(0,0)$ we get the following non-evolutionary
integrable PDE:
\begin{equation}\label{20121103:eq6}
\bigg(\frac1{u_x}
\Big(\frac{u_{tx}}{u_x}\Big)_x\bigg)_x
=
\frac1{u_x}\Big(
\frac{c_0+c_1u+c_2u^2}{u_x}
\Big)_x
\,,
\end{equation}
while for $c_1=1$ and $c_2=c_3=0$, we obtain the following integrable equation
\begin{equation}\label{20121103:eq7}
\Big(\frac{u_{tx}}{u_x}\Big)_x
=
\frac1{2u_x}+\gamma u_x
\,,
\end{equation}
where $\gamma$ is a constant.
Note that, if we apply the differential substitution $v=\log(u')$ to equation \eqref{20121020:eq8} 
with $\epsilon=0$, we get equation \eqref{20121103:eq7}.

\section{NLS type integrable systems}
\label{sec:5}

Recall from \cite[Example 4.10]{DSK12} that the following
is a triple of compatible non-local Poisson structures in two differential variables $u,v$:
$$
L_1=\partial\id
\,\,,\,\,\,\, 
L_2=\left(\begin{array}{cc} 0 & -1 \\ 1 & 0 \end{array}\right)
\,\,,\,\,\,\,
L_3=\left(\begin{array}{cc} 
v\partial^{-1}\circ v & -v\partial^{-1}\circ u \\
-u\partial^{-1}\circ v & u\partial^{-1}\circ u
\end{array}\right)
\,.
$$
We want to consider two non-local Poisson structures $H$ and $K$
which are linear combinations of them:
$H=a_1L_1+a_2L_2+a_3L_3$ and $K=b_1L_1+b_2L_2+b_3L_3$,
where $a_i$'s and $b_i$'s are constants.
As in the previous section, we are only interested in integrable Lenard-magri schemes of S-type.
In particular, we assume
that the order of the pseudodifferential operator $H$
is greater that the order of $K$,
and so we consider only the case when $b_1=0$.
Note that when $b_2=0$, we get $K=b_3L_3$,
which is a degenerate pseudodifferential operator.
In this case, we cannot apply Theorem \ref{th:lmscheme}
and we do not know how to prove integrability.
Hence, we assume that $b_2=1$.
And, since we want that the order of $H$ is greater that the order of $K$,
we assume also that $a_1=1$.

In conclusion, we consider the following compatible pair of non-local Hamiltonian structures:
\begin{equation}\label{20121109:eq1}
\begin{array}{l}
\displaystyle{
H=
\partial\id
+a_2\left(\begin{array}{cc} 0 & -1 \\ 1 & 0 \end{array}\right)
+a_3\left(\begin{array}{cc} 
v\partial^{-1}\circ v & -v\partial^{-1}\circ u \\
-u\partial^{-1}\circ v & u\partial^{-1}\circ u
\end{array}\right)
\,}\\
\displaystyle{
K=
\left(\begin{array}{cc} 0 & -1 \\ 1 & 0 \end{array}\right)
+b_3\left(\begin{array}{cc} 
v\partial^{-1}\circ v & -v\partial^{-1}\circ u \\
-u\partial^{-1}\circ v & u\partial^{-1}\circ u
\end{array}\right)
\,.}
\end{array}
\end{equation}

Note that if $a_3=b_3=0$, the above pair is such that $H$ is ``local'' differential operator,
and $K$ is invertible.
In this case, the Lenard-Magri recursion conditions give
$\tint h_{-1}=0$ and $H\frac{\delta h_{n-1}}{\delta u}=K\frac{\delta h_n}{\delta u}$ for every $n\geq0$.
Namely, $\frac{\delta h_n}{\delta u}=0$ for every $n$.
Therefore, in this case,
the corresponding Lenard-Magri scheme is of finite type, and we don't get any integrable system.
Hence, we assume that $(a_3,b_3)\neq(0,0)$.

Next, we need to find minimal fractional decompositions for $H$ and $K$.
This is given by the following
\begin{lemma}\label{20121108:lem1}
We have the following minimal fractional decomposition for the operator $L_3$:
$$
\left(\begin{array}{cc} 
v\partial^{-1}\circ v & -v\partial^{-1}\circ u \\
-u\partial^{-1}\circ v & u\partial^{-1}\circ u
\end{array}\right)
=
\left(\begin{array}{cc} 
0 & -uv \\
0 & u^2
\end{array}\right)
\left(\begin{array}{cc} 
1 & 0 \\
\frac{v}{u} & \frac1u\partial\circ u
\end{array}\right)^{-1}
$$
The rational matrix pseudodifferential operator $H$ admits the fractional decomposition 
$H=AB^{-1}$ given by
\begin{equation}\label{20121109:eq2}
A=
\left(\begin{array}{cc} 
\partial-a_2\frac{v}{u} & -a_2\frac{1}{u}\partial\circ u-a_3uv \\
\partial\circ\frac{v}{u}+a_2 & \partial\circ\frac1u\partial\circ u+a_3u^2
\end{array}\right)
\,\,,\,\,\,\,
B=
\left(\begin{array}{cc} 
1 & 0 \\
\frac{v}{u} & \frac1u\partial\circ u
\end{array}\right)
\,,
\end{equation}
which is minimal for $a_3\neq0$,
while, for $a_3=0$, $H$ is a matrix differential operator.
The rational matrix pseudodifferential operator $K$ admits the fractional decomposition 
$H=CB^{-1}$ given by
with $B$ as in \eqref{20121109:eq2} and
\begin{equation}\label{20121109:eq3}
C=
\left(\begin{array}{cc} 
-\frac{v}{u} & -\frac{1}{u}\partial\circ u-b_3uv \\
1 & b_3u^2
\end{array}\right)
\,.
\end{equation}
This decomposition is minimal for $b_3\neq0$,
while, for $b_3=0$, $K=L_2$ is an invertible matrix.
\end{lemma}
\begin{proof}
Straightforward.
\end{proof}

As usual, in order to apply the Lenard-Magri scheme it is convenient to find the kernels 
of the operators $B$ and $C$:
$$
\ker B=\mc C \left(\begin{array}{c} 0 \\ \frac1u \end{array}\right)
\,\,,\,\,\,\,
\ker C=\mc C \left(\begin{array}{c} -b_3u \\ \frac1u \end{array}\right)\,.
$$

We next compute the first few steps in the Lenard-Magri scheme.
We have the following $H$ and $K$-associations:
$\tint 0\ass{H}P_0\ass{K}\tint h_0\ass{H}P_1\ass{K}\tint h_1\ass{H}P_2$,
where
\begin{equation}\label{20121109:eq4}
\begin{array}{l}
\displaystyle{
P_0=\alpha a_3\left(\begin{array}{c} -v \\ u \end{array}\right)
\,\,,\,\,\,\,
\tint h_0=\frac{1}{2}\tint (u^2+v^2)
\,,} \\
\displaystyle{
P_1= \left(\begin{array}{c} u'-a_2v \\ v'+a_2u \end{array}\right)
\,\,,\,\,\,\,
\tint h_1=
\int\Big(
uv'+\frac{a_2}{2}(u^2+v^2)+\frac{b_3}{8}(u^2+v^2)^2
\Big)
\,,} \\
\displaystyle{
P_2=
\left(\begin{array}{c} 
v''+2a_2u'-a_2^2v+\frac{b_3}{2}\big(u(u^2+v^2)\big)^\prime +\frac{a_3-a_2b_3}{2}v(u^2+v^2) \\ 
-u''+2a_2v'+a_2^2u+\frac{b_3}{2}\big(v(u^2+v^2)\big)^\prime -\frac{a_3-a_2b_3}{2}u(u^2+v^2)
\end{array}\right)
\,.}
\end{array}
\end{equation}
Indeed, we have
$P_0=AF_0$, $BF_0=0$, for $F=\alpha\left(\begin{array}{c} 0 \\ \frac1u \end{array}\right)$.
We have
$P_0=CF_1$, $\frac{\delta h_0}{\delta u}=BF_1$, 
for $F_1=\left(\begin{array}{c} u \\ \frac{\beta}u \end{array}\right)$,
where $\alpha,\beta\in\mc C$ are chosen so that $\alpha a_3-\beta b_3=1$
(we can always do so, since, by assumption, $(a_3,b_3)\neq(0,0)$).
We have
$P_1=AF_2$, $\frac{\delta h_0}{\delta u}=BF_2$, 
for $F_2=\left(\begin{array}{c} u \\ 0 \end{array}\right)$.
We have
$P_1=CF_3$, $\frac{\delta h_1}{\delta u}=BF_3$, 
for $F_3=\left(\begin{array}{c} 
v'+a_2u+\frac{b_3}{2}u(u^2+v^2) \\ 
-\frac12\frac{u^2+v^2}{u} 
\end{array}\right)$.
And, finally, we have
$P_2=AF_4$, $\frac{\delta h_1}{\delta u}=BF_4$, 
for $F_4=F_3$.

Next, we check that the orthogonality conditions \eqref{20120621:eq2} hold for $N=0$.
We have
$F=\left(\begin{array}{c} f \\ g \end{array}\right)\in\frac{\delta h_0}{\delta u}^\perp$
if and only if $\tint (uf+vg)=0$,
namely if $f=-\frac{v}{u} g+\frac{h'}{u}$, for some $h\in\mc V$.
But in this case
$$
F=\left(\begin{array}{c} -\frac{v}{u} g+\frac{h'}{u} \\ g \end{array}\right)
=C\left(\begin{array}{c} g+b_3uh \\ -\frac{h}{u} \end{array}\right)\in\im C\,,
$$
proving the first orthogonality condition \eqref{20120621:eq2}.
As for the first orthogonality condition, if $a_3=0$ there is nothing to prove since $H$ is a matrix 
differential operator (i.e. the denominator is $\id$ in its minimal fractional decomposition).
If $a_3\neq0$, we can choose $\alpha=\frac1{a_3}$, and we have
$F=\left(\begin{array}{c} f \\ g \end{array}\right)\in P_0^\perp$
if and only if $\tint (-vf+ug)=0$,
namely if $g=\frac{v}{u} f+\frac{h'}{u}$, for some $h\in\mc V$.
But in this case
$$
F=\left(\begin{array}{c} f \\ \frac{v}{u} f+\frac{h'}{u} \end{array}\right)
=B\left(\begin{array}{c} f \\ \frac{h}{u} \end{array}\right)\in\im B\,,
$$
proving the second orthogonality condition \eqref{20120621:eq2}.
Therefore, by Theorem \ref{th:lmscheme},
we deduce that the elements \eqref{20121109:eq4}
can be extended to infinite sequences $\{\tint h_n\}_{n\in\mb Z_+}$, $\{P_n\}_{n\in\mb Z_+}$,
such that 
$\tint h_{n-1}\ass{H}P_n\ass{K}\tint h_n$.

Finally, 
we have $|H|=1$, $|K|=0$, $\dord(A)=2$, $\dord(B)=\dord(C)=\dord(D)=1$, and $\dord(P_2)=2$.
Hence, the inequality \eqref{20120911:eq1} holds.
Therefore, by Proposition \eqref{20120910:prop}
we have $\dord(P_n)=\dord(\frac{\delta h_n}{\delta u})=n$ for every $n\in\mb Z_+$.
In particular, all the elements $\tint h_n$'s and $P_n$'s are linearly independent.

In conclusion, each equation of the hierarchy $\frac{du}{dt_n}=P_n$
is integrable, and the local functionals $\tint h_n$'s are their integrals of motion.
The first ``non-trivial'' equation of this hierarchy is for $n=2$.
Letting $a_2=0$, $a_3=2\alpha$, and $b_3=2\beta$, it has the form
\begin{equation}\label{20121109:eq5}
\left\{\begin{array}{l}
\displaystyle{
\frac{du}{dt}=
v'' +\alpha v(u^2+v^2) +\beta \big(u(u^2+v^2)\big)^\prime
} \\ 
\displaystyle{
\frac{dv}{dt}=
-u'' -\alpha u(u^2+v^2) +\beta \big(v(u^2+v^2)\big)^\prime
}
\end{array}\right.
\end{equation}
If we view $u$ and $v$ as real valued functions, and we consider the complex valued function
$\psi=u+iv$, the system \eqref{20121109:eq5} can be written as the following PDE:
$$
i\frac{d\psi}{dt}=\psi''+\alpha\psi|\psi|^2+i\beta(\psi|\psi|^2)^\prime\,,
$$
which, for $\beta=0$, is the well-known Non-Linear Schroedinger equation
(see e.g. \cite{TF86,Dor93,BDSK09,DSK12}).

It is not difficult to show that when going back
the Lenard-Magri scheme is ``blocked'' at $\tint h_{-2}$ when $a_3=0$,
and it is of finite type when $a_3\neq0$.
Hence, we don't get any non-evolutionary PDE in this case.



\end{document}